\documentclass[%
onecolumn, 11pt,  tightenlines,
 aps,
pra,
]{revtex4-2}

\usepackage[ruled,lined]{algorithm2e}
\usepackage{algpseudocode}
\usepackage{graphicx}
\usepackage{dcolumn}
\usepackage{bm}
\usepackage{braket}
\usepackage{mathrsfs}
\usepackage{amssymb,amsmath,amsthm}

\newtheorem{theorem}{Theorem}

\usepackage{mathtools}

\theoremstyle{definition}

\newtheorem{innercustomthm}{Conjecture}
\newenvironment{customthm}[1]
  {\renewcommand\theinnercustomthm{#1}\innercustomthm}
  {\endinnercustomthm}
  
  \newtheorem{innercustomdef}{Definition}
\newenvironment{customdef}[1]
  {\renewcommand\theinnercustomdef{#1}\innercustomdef}
  {\endinnercustomthm}

\begin{document}

\preprint{Article}

\title{K-sparse Pure State Tomography with Phase Estimation}
\thanks{This work was supported by TUBITAK (The Scientific and Technical Research Council of Turkey) under Grant $\#$119E584.}%

\author{Burhan Gulbahar}
 \email{burhan.gulbahar@yasar.edu.tr}
\affiliation{%
 Department of Electrical and Electronics Engineering, Yasar University, Izmir, Turkey \\
}%

\date{\today}

\begin{abstract}
Quantum state tomography (QST) for reconstructing pure states requires exponentially increasing resources and measurements with the number of qubits by using state-of-the-art quantum compressive sensing (CS) methods. In this article, QST reconstruction for any pure state  composed of the superposition of $K$ different computational basis states of $n$ qubits in a specific measurement set-up, i.e., denoted as \textit{$K$-sparse}, is achieved without any initial knowledge  and with quantum polynomial-time complexity of resources based on the assumption of the existence of polynomial size quantum circuits for implementing exponentially large powers of a specially designed unitary operator. The algorithm includes $\mathcal{O}(2 \, / \, \vert c_{k}\vert^2)$ repetitions of conventional phase estimation algorithm depending on the probability $\vert c_{k}\vert^2$ of the least possible basis state in the superposition and $\mathcal{O}(d \, K \,(log K)^c)$  measurement settings with conventional quantum CS algorithms independent from the number of qubits while dependent on $K$ for constant $c$ and $d$. Quantum phase estimation algorithm is exploited based on the favorable eigenstructure of the designed operator to represent any pure state as a  superposition of eigenvectors. Linear optical set-up is presented for realizing the special unitary operator which includes beam splitters and phase shifters where propagation paths of single photon are  tracked with which-path-detectors. Quantum circuit implementation is provided by  using only CNOT, phase shift and  $-  \pi  \, / \, 2$ rotation gates around X-axis in Bloch sphere, i.e., $R_{X}(- \pi \, / \, 2)$, allowing to be realized in NISQ devices.  
Open problems are discussed regarding the existence of the unitary operator and its practical circuit implementation. 
\end{abstract}

\maketitle

\section{Introduction}
\label{sec1}

Quantum state tomography (QST) determines an unknown state by making measurements on identical copies with  highly important applications in various areas of quantum technologies \cite{d2003quantum, cramer2010efficient, christandl2012reliable}.  QST and process estimation tasks for  characterizing and reconstructing purposes generally require exponential amount of measurements \textit{$\mathcal{O}(2^{2n})$} with the number of qubits $n$ without any knowledge about the state. Low-rank density matrices with rank $r \ll 2^n$ approximating pure states are reconstructed with \textit{$\mathcal{O}(r \, 2^n \, n^c)$} Pauli measurements for a constant $c \in [2, 6]$ by exploiting compressive sensing (CS) with convex or non-convex programming approaches \cite{gross2010quantum, kyrillidis2018provable, flammia2012quantum} and experimental studies \cite{steffens2017experimentally}. Linear optical methods are already utilized such as in \cite{banchi2018multiphoton} for multi-mode multi-photon states where finite number of linear optical interferometer configurations are used. The reconstruction method requires  $\sharp$P-hard matrix permanent calculations with \textit{$\mathcal{O}\big(poly(D_{n,M},\, 2^n) \big)$} complexity where $n$ is the number of photons, $M$ is the number of modes and    $D_{n,M} = {n \choose  n-M}$. 

In this article,  QST of any pure state of $n$ qubits composed of the superposition of $K$ different computational basis states in a specific measurement set-up, i.e., denoted as \textit{$K$-sparse} pure state, is achieved in quantum polynomial-time without any knowledge about the state including the value of $K$ based on assumptions about the implementation of the black-boxes of specially designed unitary operator $U_{\vec{\Phi}}$. It is assumed that exponentially large powers of $U_{\vec{\Phi}}$, i.e., controlled-$U_{\vec{\Phi}}^{2^j}$ operations for finite $j$, can be implemented with polynomial size quantum circuits. Then, the main QST problem for $n$ qubits is basically converted to the problem of quantum CS based QST of $\log(K)$ qubits by creating independence from the number of qubits $n$ for finite $K \ll 2^n$. 

Firstly, an ancillary single qubit initialized to $(\ket{0} \, + \, \ket{1}) \, / \, \sqrt{2}$ is included by increasing the number of qubits to $n+1$. After applying  n-qubit Hadamard transform to the pure state, QST problem for the resulting $n\,+\,1$ qubits is shown to be equivalent to conventional phase estimation problem for estimating eigenvalues and projecting onto eigenvectors of $U_{\vec{\Phi}}$.  Here, we mainly exploit favorable structure of the eigenvectors of $U_{\vec{\Phi}}$ to represent any input  pure state as a superposition of eigenvectors of $U_{\vec{\Phi}}$ with the help of ancillary qubit. Then, the application of conventional phase estimation algorithm determines unknown computational basis states composing the pure state. After learning the  basis locations of sparsity, i.e., the unknown $K$ different computational basis states, the conventional quantum CS methods are applied to estimate $K$ different complex superposition coefficients.  

We present a linear optical set-up to realize such a unitary operator. It is motivated by the target of tracking the evolution of single photons through consecutive beam splitters (BSs) and phase shifters by using consecutive which-path-detectors (WPDs) after each BS. We exploit the surprising eigenstructure of the designed set-up. WPDs allow to track the evolution of photons in various settings \cite{englert1996fringe}. In addition, quantum circuit implementation of $U_{\vec{\Phi}}$ is presented to be realized in universal quantum computers or noisy-intermediate scale quantum (NISQ) devices.

More specifically, consider the following state in the problem definition composed of the superposition of $K$ different computational basis states of $n$ qubits, i.e., denoted as \textit{ $K$-sparse} pure state, for estimation with QST. The problem is defined as follows.

\begin{customdef}{1} \textbf{$K$-sparse pure state reconstruction problem}: 
\textit{Estimate the unknown values and reconstruct any pure state being in a superposition of $K$ different computational basis states of $n$ qubits defined as follows:
\begin{eqnarray}
\label{purestate}
\ket{\Psi} =   \sum_{k = 1}^{K}  c_k \, e^{\imath \, \vartheta_k}\, \ket{s_{k,n} \, s_{k, n-1} \, \hdots s_{k,1}}
 \end{eqnarray} 
where $c_k > 0$ and $c_k \, e^{\imath \, \vartheta_k}$ is the complex superposition coefficient of the unknown computational basis state $\ket{s_{k,n} \, s_{k, n-1} \, \hdots s_{k,1}}$ with unknown values of $s_{k,j} = 0$ or $1$ for $j \in [1, n]$ and $K \ll 2^n$.  The unknown values to be estimated are $c_k$ and $ \vartheta_k$ for $k \in [1, K]$, the number of superposition components $K$ and $s_{k,j}$ for $j \in [1,n]$, i.e., the locations of sparsity, without any initial knowledge about the state. The structure of the computational basis states for each qubit, i.e., $\ket{0}$ and $\ket{1}$, is initially defined depending on the measurement set-up. }
\end{customdef}

Observe that the sparsity is defined with respect to the chosen computational basis states of each qubit, i.e., $\ket{0}$ and $\ket{1}$, in the measurement set-up.   The computational basis state of each qubit depends on the given measurement set-up and the pure state to be estimated is assumed to have finite number of superposition components without knowing the exact value of $K$.  

In classical CS, the problem is defined in  \cite{baraniuk2007compressive} as finding the minimum  $\mathscr{L}_1$-norm with $\widehat{\mathbf{x}} = \mbox{argmin} \, \Vert \mathbf{x'} \Vert_1$ such that  $\mathbf{y} = \Theta \, \mathbf{x'}$ where $\mathbf{y}$ is the measurement result, $\mathbf{x}$ is the unknown vector and $\Theta $ is a Gaussian matrix. It is observed that if $\mathbf{x}$ is $K$-sparse, then $\Theta $ with the dimension $M \times N$ is enough where  $M \, = \, c\, K \, log(N \, / \, K)$ is the number of measurements with constant $c$ and $N$ is the dimension of the signal $\mathbf{x}$. The computational complexity of the solution with basis pursuit is $\mathcal{O}(N^3)$. In  quantum CS, e.g., with convex solution \cite{gross2010quantum}, randomly chosen Pauli expectations $tr(\mathbf{P}^k \,\rho)$  are utilized to minimize $\Vert \sigma \Vert_{tr}$ where $ tr(\mathbf{P}^k\, \sigma ) = tr(\mathbf{P}^k\, \rho )$ where $\Vert \sigma \Vert_{tr}$  denotes the sum of the singular values of $\sigma $,  the unknown density matrix is  $\rho$, $M$  is the number of measurements with $k \in [1, M]$ and $\mathbf{P}^k  \, \equiv \, \bigotimes_{j=1}^{n} P_i^k$ is a random Pauli measurement  with $P_i^k \in \lbrace \mathtt{1}, \sigma^x, \,\sigma^y, \, \sigma^z \rbrace$. $M$ is mainly reduced from the requirement of roughly $2^{2n}$ measurements to $2^n$ measurements by using CS based on the assumption of the purity of the sampled state.
 
There are diverse number of studies utilizing quantum CS or alternative approaches in more practical manners for further reducing the number of measurements and the required resources such as online learning and shadow tomography \cite{aaronson2019online, aaronson2019shadow}, self-calibrating quantum state tomography by relaxing the blind tomography problem to sparse de-mixing \cite{roth2020semi} and hierarchical compressed sensing \cite{eisert2021hierarchical}, adaptive compressive tomography without  a-priori information  \cite{ahn2019adaptive},  reduced density matrices \cite{cotler2020quantum, xin2017quantum}, matrix product state tomography \cite{lanyon2017efficient, cramer2010efficient}, neural network  \cite{torlai2018neural}, machine learning \cite{lohani2020machine} based approaches or different methods including \cite{pereira2021scalable}.
There is not any polynomial-time quantum algorithm or solution method available for exactly reconstructing $K$-sparse pure states  in (\ref{purestate}), i.e., $K$-sparse vectors in dimension $2^n$, with neither quantum nor classical polynomial-time resource complexity.   It is an open issue to achieve $\mathscr{L}_1$-norm minimization based QST of $K$-sparse pure states with quantum algorithms of polynomial-time complexity.

In this article, we provide an alternative quantum polynomial-time solution with $\mathcal{O}(d \, K \,(log K)^c)$  measurement settings for $K$-sparse pure states in parallel with  $\mathcal{O}\big(c\, K \, log(2^n \, / \, K) \big)$ measurements in classical CS without any complexity exponentially growing with $2^n$. On the other hand, the computational complexity is also maintained as quantum polynomial-time providing a complete practicality in QST tasks in analogy with the practical utilization of CS in classical world with exponentially smaller dimensions compared with the Hilbert space size of quantum states.  We provide the locations of the $K$-sparse points without solving any $\mathscr{L}_1$-norm minimization problem with a different perspective compared with conventional CS. Phase estimation presents a new approach for QST of pure states while encouraging the design of new unitary operators $U_{\vec{\Phi}}$ with the proposed eigenstructure and practical implementation capability. It is an open issue to realize polynomial size quantum circuit implementations for exponentially large powers of the presented specific design based on linear optics  and WPDs. Furthermore, analyzing the effects of noise in input state and estimation errors are important for practical considerations \cite{riofrio2017experimental, steffens2017experimentally}. Extension to mixed state inputs is also an open issue.

The proposed architecture is promising to be utilized in all state estimation  and process modeling tasks \cite{d2003quantum, cramer2010efficient, christandl2012reliable, aaronson2019online}. Quantifying the amount of entanglement existing in a quantum state  is another potential application \cite{schneeloch2019quantifying}. Quantification generally requires full state tomography and complex calculations for entanglement monotones \cite{di2013embedding}. Furthermore, the proposed method can be utilized to map classical states into eigenvectors of $U_{\vec{\Phi}}$  for various machine learning procedures similar to embedding into the amplitude, basis or the dynamics of quantum systems through Hamiltonian embedding \cite{schuld2018supervised}, or using various quantum feature map operators  \cite{havlivcek2019supervised, schuld2019quantum, lloyd2020quantum}. The embedded classical data is extracted reliably in quantum polynomial-time by using the proposed QST architecture. 

One fundamental theorem and supporting conjecture are formulated in this article. Theorem-1 reduces the exponential amount of resources and measurements necessary in conventional QST and quantum CS algorithms to quantum polynomial-time resources as follows:

\begin{theorem}
There exists a unitary operator $U_{\vec{\Phi}}$ with a specially defined eigenstructure to be utilized in quantum phase estimation algorithm  so that any $K$-sparse $n$-qubit pure state  can be reconstructed  after $\mathcal{O}(1 \, / \, m)$ repetitions of conventional $t$-bit quantum phase estimation algorithm with  $\widetilde{t} \, \approx t\,+\,log \big(2 \, + \, 1 \, / \, \,(2 \, \epsilon) \big)$ ancillary qubits having success probability of at least $(1  \, - \, \epsilon)$ and consecutive $\mathcal{O}(d \, K \,(log K)^c)$ measurement settings based on existing quantum CS methods for $c \in [2, 6]$ and some constant $d$ reducing the error exponentially. $m$ is such that the number of repetitions is $\mathcal{O}(2 \, / \, \min_{k \in [1, K]} \lbrace \vert c_k \vert^2 \rbrace)$ independent from the number of qubits while depending on the probability $\vert c_k \vert^2$  of the least probable basis state in the superposition for $k \in [1, K]$.   It is assumed that black-boxes of $U_{\vec{\Phi}}^{2^j}$ are available for $j \in [0, \widetilde{t}-1]$.
\end{theorem}

\begin{proof}
The proof is provided in Section \ref{sec2}.
\end{proof}

The supporting conjecture proposes that unitary operators $U_{\vec{\Phi}}$ with the special eigenstructure exploited in Theorem-1 can be realized by using polynomial size quantum circuits and linear optics.  

\begin{customthm}{1}   
\textit{There exist a unitary operator $U_{\vec{\Phi}}$ and its polynomial size quantum circuit implementation composed of CNOT gates, $R_{X}(\frac{-\pi}{2}) \equiv  \frac{1}{\sqrt{2}}\begin{bmatrix}1 & \imath\\ \imath & 1\end{bmatrix} $ gates  and phase shifters $\Phi_{j} \, \equiv   \, \begin{bmatrix}
1 &   0 \\
0 & e^{\imath \, \phi_{j}}\\
\end{bmatrix}$  for $j \in [1, n+1]$ such that it has distinct eigenvalues and unique pairs of eigenvectors with the form specified as quantum states $ H^{\otimes n} \,\ket{s_{k,n}  \, s_{k,n-1} \, \hdots s_{k,1}} \, \big( \alpha_{k,0} \, \ket{0} \, + \, \alpha_{k,1} \, \ket{1}  \big)$  or $ H^{\otimes n} \,\ket{s_{k,n}  \, s_{k,n-1} \, \hdots s_{k,1}} \, \big( \beta_{k,0} \, \ket{0} \, + \, \beta_{k,1} \, \ket{1}  \big)$ corresponding to each $\ket{s_{k,n}  \, s_{k,n-1}  \, \hdots s_{k,1}}$ for $k \in [1, 2^n]$. The parameters of a pair of eigenvectors satisfy the relation such that the parameters of the first eigenvector are $\alpha_{k,0}  = a_0 \, + \, \imath \, b_0$ and $\alpha_{k,1} = a_1$ with $a_0$, $b_0$ and $a_1 \, \in \mathcal{R}$ and $a_0^2 \, + \, b_0^2 \, + \, a_1^2 = 1$ and the parameters of the other eigenvector are $\beta_{k,0} = a_1$ and $ \beta_{k,1} = - a_0 \, + \, \imath \, b_0$  (or their multiplication with arbitrary phase). The parameters $\alpha_{k,0}$,   $\alpha_{k,1}$,  $\beta_{k,0}$ and  $\beta_{k,1}$ depend on $\phi_{j}$  for $j \in [1, n+1]$, $H^{\otimes n}$ is n-qubit Hadamard transform and  $\ket{s_{k,n}  \, s_{k,n-1}  \, \hdots s_{k,1}}$ denotes computational basis state with $s_{k,j} = 0$ or $1$ for $j \in [1, n]$. The existence of the operator satisfying the specific form of eigenstructure depends on the chosen set $\phi_j$ for $j \in [1, n+1]$.}
\end{customthm}

Theoretical analysis and supporting evidence are provided in Section \ref{sec5a}. Linear  optical set-up is provided in Section \ref{sec3} to realize such unitary operators where it is motivated by consecutive WPDs tracking the evolution of  a single photon in a linear optical set-up. In Section \ref{sec4}, quantum circuit implementation is provided for the set-up. Numerical  analysis shows that if uniformly distributed phase shifts are utilized in the linear optical design or the proposed quantum circuit implementation, then $U_{\vec{\Phi}}$ with the specific form of eigenstructure is obtained. However, it is an open issue to determine the conditions on the phase shift values under which  the proposed eigenstructure is satisfied. This is the reason that we provide the result as a conjecture.

The remainder of the paper is organized as follows. In Section \ref{sec2}, QST algorithm is presented. In Section \ref{sec3}, linear optical set-up providing a design of the targeted $U_{\vec{\Phi}}$ is presented. Then, in Section \ref{sec4}, quantum circuit implementation for the designed operator is provided. Eigenstructure of $U_{\vec{\Phi}}$ is analyzed in Section \ref{sec5}. Open problems are presented in Section \ref{sec6}.

\section{$K$-sparse Pure State Tomography Algorithm}
\label{sec2}

QST problem is converted to eigenvector estimation problem by using the eigenvalues estimated with phase estimation algorithm \cite{kitaev1995quantum, nielsen2010quantum}.  The algorithm is summarized in Algorithm-\ref{alg:QST} with two important phases. The proof of Theorem-1 and the description of the algorithm are provided next.

\begin{algorithm}
\label{alg:QST}
\caption{$K$-sparse pure state tomography with phase estimation}
\textbf{$1^{st}$ PHASE:}
\begin{enumerate}
\item Initial pure state: $\ket{\Psi} =   \sum_{k = 1}^{K}  c_k \, e^{\imath \, \vartheta_k}\, \ket{s_{k,n} \, s_{k, n-1} \, \hdots s_{k,1}}$ with unknown $K$, $c_k$, $\varphi_k$ and  $\ket{s_{k,n} \, s_{k, n-1} \, \hdots s_{k,1}}$ for $k \in [1, K]$\\
\item Add an ancillary state $(\ket{0} \, + \, \ket{1}) \, / \, \sqrt{2}$ and apply n-qubit Hadamard transformation to $\ket{\Psi}$ by transforming into a superposition of eigenvectors of the unitary operator $U_{\vec{\Phi}}$: 
$$
\ket{\Psi_e} =  \sum_{k = 1}^{K} c_k \, e^{\imath \, \vartheta_k}\,  H^{\otimes n} \,\ket{s_{k,n} \,  \, s_{k, n-1} \, \hdots s_{k,1}} \bigg(\frac{\ket{0} \, + \, \ket{1}}{\sqrt{2}} \bigg) \,  =   \, \sum_{k = 1}^{K} c_k \, e^{\imath \, \vartheta_k}\, \bigg( \sum_{l = 1}^2 d_{k,l}  \, \ket{E_{k,l}}  \bigg)
$$\\
\item Apply phase estimation and collapse the state to $\ket{\widetilde{\lambda_{k,1}}} \, \ket{E_{k,l}}$ for $l  \in [1,2]$ with probabilities $\vert c_{k}\vert^2 \, \vert d_{k,1}\vert^2  = \vert c_{k}\vert^2 \, (1/2 \, + \,a_{k,0} \, a_{k,1} )$ and $\vert c_{k}\vert^2 \, \vert d_{k,2}\vert^2  = \vert c_{k}\vert^2 \, (1/2 \, - \,a_{k,0} \, a_{k,1} )$. \\
\item Apply n-qubit Hadamard transformation to the second register except the ancillary qubit converting the measured state to either $ \ket{\widetilde{\lambda_{k,1}}} \,  \ket{s_{k,n} \, s_{k, n-1} \, \hdots s_{k,1}} \, \big( \alpha_{k,0} \, \ket{0} \, + \, \alpha_{k,1} \, \ket{1}  \big) $ or  $ \ket{\widetilde{\lambda_{k,2}}} \,  \ket{s_{k,n} \, s_{k, n-1} \, \hdots s_{k,1}} \,  \big( \beta_{k,0} \, \ket{0} \, + \, \beta_{k,1} \, \ket{1}  \big)$.\\
\item Measure the second register except the ancillary qubit to obtain $\ket{s_{k,n} \, s_{k, n-1} \, \hdots s_{k,1}}$.\\
\item Perform the first five steps  $\mathcal{O}(2 \, / \, \min_{k \in [1, K]} \lbrace \vert c_k \vert^2 \rbrace)$ times completing the estimation of $K$ and $\ket{s_{k,n} \, s_{k, n-1} \, \hdots s_{k,1}}$ for $k \in [1, K]$.\\
\end{enumerate}
\textbf{$2^{nd}$ PHASE:}
\begin{enumerate}
\item $\ket{\Psi}$ is estimated by converting the problem to QST problem of  $log(K)$-qubit pure states by using the projection operators   $P_{k}   \equiv \ket{s_{k,n} \, s_{k, n-1} \,  \hdots s_{k,1}} \bra{s_{k,1} \,  \hdots \, s_{k, n-1}  s_{k,n}} $ and $I - P_k$ since the states $\ket{s_{k,n} \, s_{k, n-1} \, \hdots s_{k,1}}$ for $k \in [1, K]$ are obtained in the first phase of the algorithm. Existing quantum CS methods are utilized in this step which require $\mathcal{O}(d\, K \,(log K)^c)$ measurement settings for $c \in [2, 6]$  and constant $d$ while completing the calculation of $c_k \, e^{\imath \, \varphi_k}$ for $k \in [1, K]$.
\end{enumerate}
\end{algorithm}

\begin{proof}[\textbf{Proof of Theorem-1:}]
Assume that there exists a special unitary transform $U_{\vec{\Phi}}$ with  a favorable eigenstructure as described next. In the first phase, the unknown state  $\ket{\Psi}$ is firstly converted to the following by using n-qubit Hadamard transformation and an ancillary qubit to encode as the superposition of the eigenvectors of $U_{\vec{\Phi}}$:
\begin{eqnarray}
\ket{\Psi_e}  \,  & = & \,   \sum_{k = 1}^{K} c_k \, e^{\imath \, \vartheta_k}\,  H^{\otimes n} \,\ket{s_{k,n} \,   \hdots s_{k,1}} \, \bigg(\frac{\ket{0} \, + \, \ket{1}}{\sqrt{2}} \bigg)  \\
   & = & \, \sum_{k = 1}^{K} c_k \, e^{\imath \, \vartheta_k}\, H^{\otimes n} \,\ket{s_{k,n}  \, \hdots s_{k,1}} \,  \bigg( d_{k,1}  \big( \alpha_{k,0} \, \ket{0} \, + \, \alpha_{k,1} \, \ket{1}  \big) + d_{k,2}  \big( \beta_{k,0} \, \ket{0} \, + \, \beta_{k,1} \, \ket{1}  \big) \bigg) \hspace{0.2in}   \\
    & = & \, \sum_{k = 1}^{K} c_k \, e^{\imath \, \vartheta_k}\, \big( d_{k,1}  \ket{E_{k,1}}\, + \, d_{k,2}  \, \ket{E_{k,2}}\big) 
 \end{eqnarray}  
where  $\ket{E_{k,1}}$ and $\ket{E_{k,2}}$ are eigenvectors of $U_{\vec{\Phi}}$, $ \lbrace d_{k, 1}, \,d_{k, 2} \rbrace \in \mathbb{C}$, $\alpha_{k,0} \, \equiv \, a_{k, 0} \, + \, \imath \, b_{k, 0}$, $\alpha_{k,1} = a_{k, 1}$ for  the parameters $ \lbrace a_{k,0}, \, a_{k,1}, \, b_{k,0}  \rbrace \in \mathbb{R}$,  and it is assumed that there is a special relation between the values of $\alpha_{k,j}$ and $\beta_{k,j}$ for $j \in [0, \,1]$, i.e., denoted as duality relation, defined as $\beta_{k,0} =  \, a_{k, 1}$ and $\beta_{k,1} = \, - \, a_{k, 0} \,   + \, \imath \, b_{k, 0}$ (or multiplication of $\beta_{k,0}$ and $\beta_{k,1}$ with arbitrary phase, e.g., $-1$). All the parameters depend on phase shift values in the design of  $U_{\vec{\Phi}}$ where eigenvectors $\ket{E_{k,1}}$ and $\ket{E_{k,2}}$ are defined as follows:
\begin{eqnarray}
\ket{E_{k,1}} \,  & \equiv & \,  H^{\otimes n} \,\ket{s_{k,n}  \, \hdots s_{k,1}} \, \big( \alpha_{k,0} \, \ket{0} \, + \, \alpha_{k,1} \, \ket{1}  \big) \\
\ket{E_{k,2}} \,  & \equiv  &  \,  H^{\otimes n} \,\ket{s_{k,n}  \, \hdots s_{k,1}} \, \big( \beta_{k,0} \, \ket{0} \, + \, \beta_{k,1} \, \ket{1}  \big) 
 \end{eqnarray} 

We assumed that $U_{\vec{\Phi}}$ has a special eigenstructure with the pair of eigenvectors $\ket{E_{k,1}}$ and $\ket{E_{k,2}}$ corresponding to each unique $\ket{s_{k,n} \,  s_{k,n-1} \, \hdots s_{k,1}}$ such that these eigenvectors  are separable pure states and they differ only in terms of the state of the first qubit, i.e., as $\big( \alpha_{k,0} \, \ket{0} \, + \, \alpha_{k,1} \, \ket{1}  \big)$ or $\big( \beta_{k,0} \, \ket{0} \, + \, \beta_{k,1} \, \ket{1}  \big)$. Furthermore, our assumption includes the special duality relation between the values of $\alpha_{k,j}$ and $\beta_{k,j}$ for $j \in [0, \,1]$ as described. In  addition, it is assumed that all the eigenvalues of $U_{\vec{\Phi}}$  are different which will be exploited during phase estimation as described next.
 
\begin{figure}[ht!]
\includegraphics[ width=4.5in]{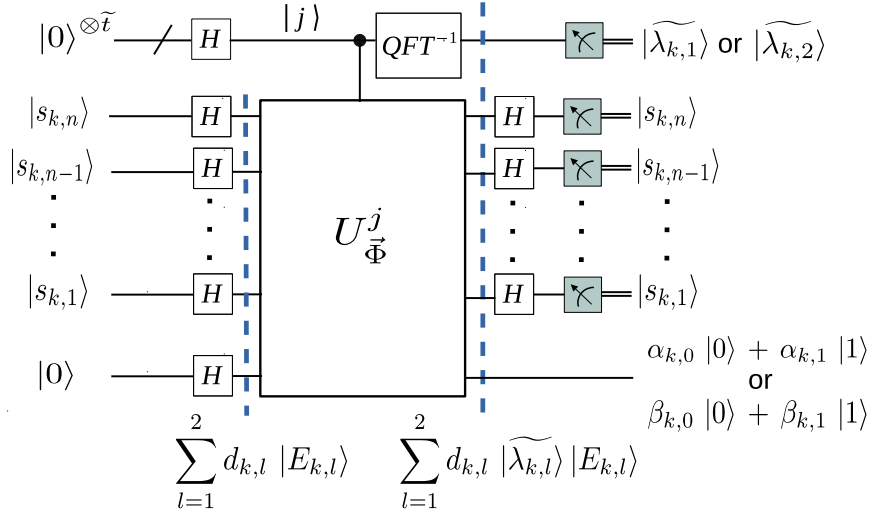}
\caption{Estimation of a single computational basis state $\ket{s_{k,n} \, s_{k, n-1} \, \hdots s_{k,1}} $ by using phase estimation algorithm. Ancillary qubit state $(\ket{0} \, + \, \ket{1})  \, / \,  \sqrt{2} $ encodes $ H^{\otimes n} \,\ket{s_{k,n} \,  s_{k,n-1} \, \hdots s_{k,1}} \, \big( (\ket{0} \, + \, \ket{1})  \, / \,  \sqrt{2} \big) $ as the superposition of two eigenvectors of $U_{\vec{\Phi}}$, i.e., $\ket{E_{k,1}}$ and $\ket{E_{k,2}}$, with  $ H^{\otimes n} \,\ket{s_{k,n} \,  s_{k,n-1} \, \hdots s_{k,1}} \, \big( (\ket{0} \, + \, \ket{1})  \, / \,  \sqrt{2} \big)$ being equal to $\sum_{l =1}^2 d_{k,l} \ket{E_{k,l}}$ where $\ket{E_{k,1}} \,   \equiv   \,  H^{\otimes n} \,\ket{s_{k,n}  \, s_{k,n-1}  \, \hdots s_{k,1}} \, \big( \alpha_{k,0} \, \ket{0} \, + \, \alpha_{k,1} \, \ket{1}  \big)$ and
$ \ket{E_{k,2}} \,   \equiv    \,  H^{\otimes n} \,\ket{s_{k,n}  \,s_{k,n-1}  \, \hdots s_{k,1}} \, \big( \beta_{k,0} \, \ket{0} \, + \, \beta_{k,1} \, \ket{1}  \big)$. In the proposed QST algorithm, the state $\ket{\Psi} =   \sum_{k = 1}^{K}  c_k \, e^{\imath \, \vartheta_k}\, \ket{s_{k,n} \, s_{k, n-1} \, \hdots s_{k,1}}$ is the input as the superposition of the computational basis states while the ancillary qubit encodes each computational basis state $\ket{s_{k,n} \, s_{k, n-1} \, \hdots s_{k,1}} $ as  superpositions of different pairs of eigenvectors. After the phase estimation,  the first register with $\widetilde{t}$ qubits is measured. Then, n-qubit Hadamard transformation is applied to the second register except the ancillary qubit. The state  collapses to either $ \ket{\widetilde{\lambda_{k,1}}} \, \ket{s_{k,n} \, s_{k, n-1} \, \hdots s_{k,1}} \,  \big( \alpha_{k,0} \, \ket{0} \, + \, \alpha_{k,1} \, \ket{1}  \big)  $ or  $ \ket{\widetilde{\lambda_{k,2}}} \,  \ket{s_{k,n} \, s_{k, n-1} \, \hdots s_{k,1}} \,  \big( \beta_{k,0} \, \ket{0} \, + \, \beta_{k,1} \, \ket{1}  \big) $. }
\label{fig5}
\end{figure}

The parameters $ d_{k,1}$ and $d_{k,2}$ with $\vert d_{k,1} \vert^2  \, + \, \vert d_{k,2} \vert^2$ = 1 are calculated easily if  $ \lbrace a_{k,0}, \, a_{k,1}, \, b_{k,0}  \rbrace$ are known. They satisfy the following for $\beta_{k,0} \, = \, a_{k, 1} \, \mbox{and} \,  \beta_{k,1} = - a_{k, 0} \,   + \, \imath \, b_{k, 0}$:
\begin{equation}
 d_{k,1} \, = \, \frac{\beta_{k,0}-\beta_{k,1}}{\sqrt{2}}; \,\,\,\, d_{k,2} \, = \,  
 \frac{\beta_{k,0}+\beta_{k,1}^*  }{\sqrt{2}}  
\end{equation}
We observe that $\vert d_{k,1} \vert^2 = (1/2 \,+ \,a_{k,0} \, a_{k,1} )$ and  $\vert d_{k,2} \vert^2 = (1/2 \,- \,a_{k,0} \, a_{k,1} )$.
Then, phase estimation algorithm with input $\ket{\Psi_e}$ is exploited to estimate the eigenvalues $e^{\imath \, \lambda_{k,1}}$ and $e^{\imath \, \lambda_{k,2}}$ corresponding to each pair of $\ket{E_{k,1}}$ and $\ket{E_{k,2}}$ resulting in the following state: 
\begin{eqnarray}
 \ket{\Psi_o}  \, \equiv  \,  \sum_{k = 1}^{K} c_k \, e^{\imath \, \vartheta_k}\, \big( d_{k,1} \ket{\widetilde{\lambda_{k,1}}} \, \ket{E_{k,1}}\, + \, d_{k,2}  \, \ket{\widetilde{\lambda_{k,2}}} \, \ket{E_{k,2}}\big)   
 \end{eqnarray} 
where $\ket{\widetilde{\lambda_{k,1}}} $ and $\ket{\widetilde{\lambda_{k,2}}} $ are $t$-bit approximations to the exact values of the phases of the eigenvalues which can be realized with a quantum circuit including $ \widetilde{t}  \approx t\,+\,log\big(2 \, + \, \, 1 \, / \, \,(2 \, \epsilon) \big)$ ancillary qubits in the first register as shown in Fig. \ref{fig5} with success probability of at least $(1  \, - \, \epsilon)$  \cite{nielsen2010quantum}.  On the other hand, measurement of $\ket{\Psi_o}$ collapses the state to $\ket{\widetilde{\lambda_{k,l}}} \, \ket{E_{k,l}}$ for $l  \in [1,2]$ with probabilities $\vert c_{k}\vert^2 \, \vert d_{k,1}\vert^2  = \vert c_{k}\vert^2 \, (1/2 \, + \,a_{k,0} \, a_{k,1} )$ and $\vert c_{k}\vert^2 \, \vert d_{k,2}\vert^2  = \vert c_{k}\vert^2 \, (1/2 \, - \,a_{k,0} \, a_{k,1} )$. Applying the n-qubit Hadamard transform to the second register of the resulting state, i.e., $\ket{E_{k,l}}$ for $l  \in [1,2]$, except the ancillary qubit results in either $  \ket{s_{k,n} \, s_{k, n-1} \, \hdots s_{k,1}} \, \big( \alpha_{k,0} \, \ket{0} \, + \, \alpha_{k,1} \, \ket{1}  \big) $ or  $  \ket{s_{k,n} \, s_{k, n-1} \, \hdots s_{k,1}} \, \big( \beta_{k,0} \, \ket{0} \, + \, \beta_{k,1} \, \ket{1}  \big)$ in the second register. Then, measuring the unknown qubits gives the values of $s_{k,j}$ for $j \in [1, n]$. Another important assumption is that all the eigenvalues of $U_{\vec{\Phi}}$  are different so that there will be one-to-one mapping between each $\ket{s_{k,n} \, \, s_{k, n-1} \,  \hdots s_{k,1}}$ and the specific eigenvalues $\ket{\widetilde{\lambda_{k,1}}} $ and $\ket{\widetilde{\lambda_{k,2}}}$ so that the measurement restores the corresponding  $\ket{s_{k,n} \, s_{k, n-1} \, \hdots s_{k,1}}$. In order to estimate all values for $k \in [1,K]$, we apply standard phase estimation algorithm shown in Fig. \ref{fig5} as much as $\mathcal{O}(1 \, / \, m)$ times  with the designed operator $U_{\vec{\Phi}}$ where $m$ is chosen as follows:
$$
m = \min_{k} \big \lbrace \max \lbrace \vert c_{k} \, \vert^2 \, \vert d_{k,1}\vert^2, \, \vert c_{k}\vert^2   \, \vert d_{k,2}\vert^2 \rbrace \big \rbrace = \min_{k} \big \lbrace \vert c_{k}\vert^2 \, \max \lbrace  \big(\frac{1}{2} \, + \,a_{k,0} \, a_{k,1} \big), \,   \big(\frac{1}{2} \, - \,a_{k,0} \, a_{k,1} \big)  \rbrace \big \rbrace
$$
for detecting one of the pairs of eigenvalues for each $\ket{s_{k,n} \, s_{k, n-1} \, \hdots s_{k,1}}$ where the one with higher $\vert d_{k,l}\vert^2$ has higher probability for $l \in [1, \,2]$. Since  $ \big(1 \, / \, 2 \, + \,a_{k,0} \, a_{k,1} \big) $ and $ \big(1 \, / \, 2 \, - \,a_{k,0} \, a_{k,1} \big)$ are varying in opposite directions as the multiplication $a_{k,0} \, a_{k,1}$ varies, it creates a balanced effect on the number of measurements practically independent from the number of qubits. In other words, $\max \lbrace  \big(1 \, / \, 2 \, + \,a_{k,0} \, a_{k,1} \big), \,   \big(1 \, / \, 2 \, - \,a_{k,0} \, a_{k,1} \big)  \rbrace $ is larger than $1 \, / \, 2$. Therefore,  the number of  phase estimation steps is $\mathcal{O}(2 \, / \, \min_{k \in [1, K]} \lbrace \vert c_k \vert^2 \rbrace)$ independent from the number of data qubits while depending on the probability $\vert c_k \vert^2$  of the least probable state in the superposition for $k \in [1, K]$.

Then, detected unknown basis states and the phases of eigenvalues are paired as the list $\lbrace  \widetilde{\lambda_{k,l}}, \, \ket{s_{k,n} \, \, s_{k, n-1} \, \hdots s_{k,1}}\rbrace$ by choosing the first detected $l \in [1,2]$ for  $\ket{s_{k,n} \, s_{k, n-1} \, \hdots s_{k,1}}$. Therefore,  the first phase of the algorithm is completed with estimation of $K$ and $\ket{s_{k,n} \, s_{k, n-1} \, \hdots s_{k,1}}$ for $k \in [1, K]$. Therefore, by combining the unique eigenvector structure of $U_{\vec{\Phi}}$ with the power of phase estimation, we can easily estimate the number of superposition components and unknown basis states $\ket{s_{k,n} \, s_{k, n-1} \, \hdots s_{k,1}}$ with quantum polynomial-time resources. 

On the other hand, observe that phase estimation algorithm uses black-boxes, i.e., controlled-$U_{\vec{\Phi}}^{2^j}$  for $j \in [0, \, \widetilde{t} \, - \,1]$, so that we assume that there is a polynomial size quantum circuit implementing the exponentially large powers of  controlled-$U_{\vec{\Phi}}$ operations. This is an important open issue for practical implementation and for clarifying theoretical bounds of the complexity for reconstructing $K$-sparse pure states by exploiting quantum computation, i.e., phase estimation.

In the second phase of the algorithm, $c_k \, e^{\imath \, \varphi_k}$ is estimated for $k \in [1, K]$. Since each state $\ket{s_{k,n} \, s_{k, n-1} \, \hdots s_{k,1}}$ is known, projection of $\ket{\Psi} $ onto $P_{k}\,\equiv \, \ket{0_k}\bra{0_k}$ and $P_{k}^{\perp} \, \equiv \, I \, - \, P_{k} \equiv \ket{1_k}\bra{1_k}$ converts the problem to QST for pure states of $log(K)$ qubits where $P_{k}$ is defined as follows:
\begin{equation}
 P_{k}    \equiv \ket{s_{k,n} \, s_{k, n-1} \, \hdots s_{k,1}} \bra{s_{k,1} \,  \hdots \, s_{k, n-1} \,  s_{k,n}} 
\end{equation}
The existing quantum CS methods in \cite{gross2010quantum, kyrillidis2018provable} can be utilized for  $\mathcal{O}(K \,(log K)^c)$ measurement settings for $c \in [2, 6]$ by estimating $c_k \, e^{\imath \, \varphi_k}$ for $k \in [1, K]$ with high accuracy and low error probability. In fact, as shown in \cite{gross2010quantum}, $\mathcal{O}(d\, K \,(log K)^c)$ measurement settings reduce the error exponentially with some parameter $d$ by using certified tomography for pure states.

As a result, the state is fully reconstructed  after $\mathcal{O}(2 \, / \, \min_{k \in [1, K]} \lbrace \vert c_k \vert^2 \rbrace)$ repetitions of conventional $t$-bit  quantum phase estimation algorithm to estimate $K$ and $\ket{s_{k,n} \, s_{k, n-1} \, \hdots s_{k,1}}$ for $k \in [1,\, K]$, and then $\mathcal{O}(d \, K \,(log K)^c)$  measurement settings with constant $c$ and $d$ for quantum CS to estimate $c_k \, e^{\imath \, \varphi_k}$ for $k \in [1,K]$. This completes the proof.
\end{proof}

Next, the design of the optical set-up composed of linear optical elements and WPDs, and its quantum circuit implementation are provided having the unique and favorable eigenstructure.

\section{Design of Linear Optical Set-up Realizing the Unitary Operator $U_{\vec{\Phi}}$}
\label{sec3}

A novel linear optical set-up motivated for tracking the path of a single photon through consecutive BSs and phase shifters by using consecutive WPDs surprisingly provides an interesting relation among the eigenvectors. WPDs are utilized to track the history of a single photon state in various optical architectures for exploring quantum mechanical fundamentals \cite{englert1996fringe}. Multi-WPD implementation is shown in Fig. \ref{fig1} motivating the designed set-up in this article. Single photon diffraction paths   are tracked with the help of WPDs. The analogical linear optical circuit and its combination with WPDs are realized by replacing planes with BSs, path length differences with phase shifters and the effects of WPDs with CNOT gates as shown in Fig. \ref{fig2}(a). Quantum circuit modeling for symmetric BSs is shown in Fig. \ref{fig2}(b) as described in the next section.  

\begin{figure*}[t!]
\includegraphics[width=4in]{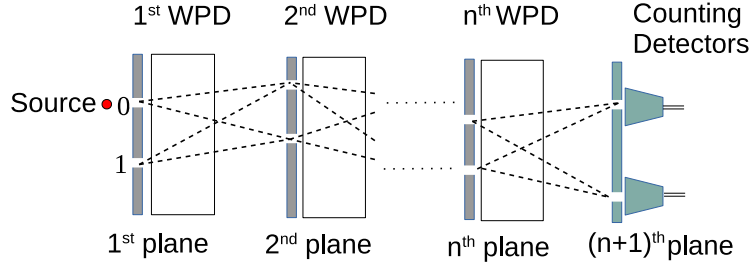} 
\caption{Multi-plane evolution of diffracting single photon tracked with WPDs motivating the designed linear optical set-up.}
\label{fig1}
\end{figure*}

\begin{figure*}[t!]
\includegraphics[width=6in]{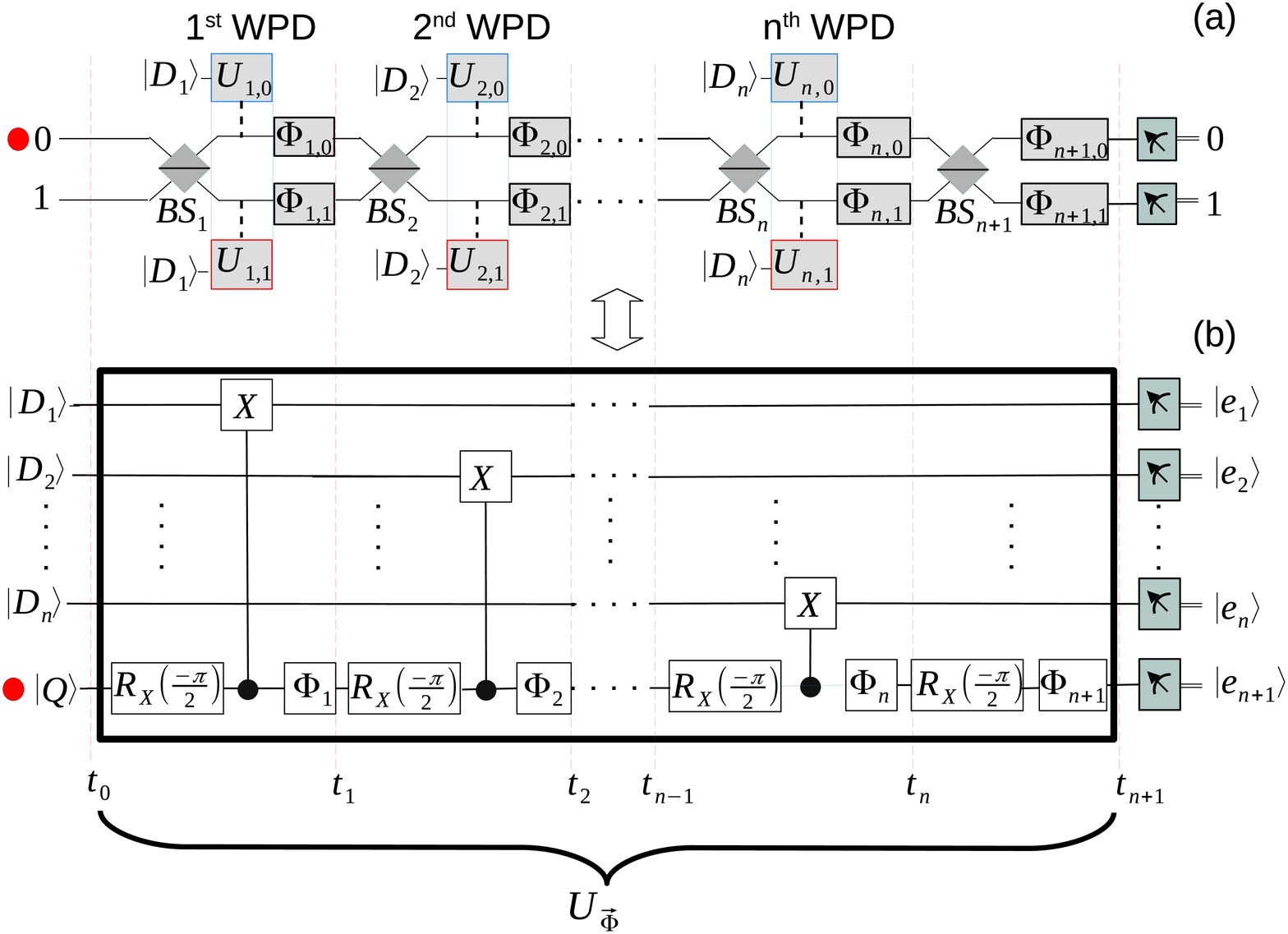} 
\caption{(a) Linear optical set-up and WPDs for realizing the targeted unitary operator $U_{\vec{\Phi}}$ where the set-up includes consecutive $n+1$ BSs, phase shifters and WPDs tracking the paths of the single photon input. $\ket{D_k}$ for $k \in [1,n]$ shows the initial states of WPDs while unitary gates $U_{k,0}$ or $U_{k,1}$ are applied on  $\ket{D_k}$ depending on the path selected by the photon. (b) Quantum circuit implementation of the set-up for symmetric BSs. $\ket{Q}$ shows the initial state of the single photon. The proposed quantum circuit realizes a unitary operator denoted with $U_{\vec{\Phi}}$.}
\label{fig2}
\end{figure*}

The paths obtained with consecutive symmetric BSs and corresponding phase shifters result in a linear optical set-up where the single photon is tracked by using consecutive WPDs with Hilbert space size of the circuit being $2^{n+1}$. It is shown in Appendix \ref{appA} that the final states of WPDs and the single photon for any initial basis state of WPDs are classically simulable. 
The final states are also derived in Appendix \ref{appA}. Defining the final state of $k$th WPD as $ \ket{d_j}_{k} \equiv U_{k, j} \ket{D_k}$ for $j \in [0, 1]$, the final state of the photon as $\ket{j_{n+1}}_{n+1}$, the combined state of the photon and WPDs as $\ket{j} \, \equiv \, \, \ket{j_{n+1}}_{n+1}\, \ket{d_{j_n}}_n \, \hdots \, \ket{d_{j_1}}_1 $ and $A(j) \equiv 
K_{j_1} \prod_{k=1}^{n} \chi_{k+1, j_{k}, j_{k+1}}$ with decimal value of $ j \equiv j_{n+1}\, 2^n \, + \, \sum_{k=1}^n d_{j_k} \, 2^{k-1}$ allows to simplify  $\ket{\Psi_{n+1}}$ as  
\begin{equation}
\ket{\Psi_{n+1}} = \sum_{j = 0}^{2^{n+1}-1} A(j) \, \ket{j} = \sum_{j = 0}^{2^{n+1}-1} \bigg( K_{j_1} \prod_{k=1}^{n} \chi_{k+1, j_{k}, j_{k+1}} \bigg) \ket{j} 
\end{equation}
where the parameters $K_{j_1}$ and $\chi_{k+1, j_{k}, j_{k+1}}$ composing $A(j)$  are defined in Appendix \ref{appA}. It is observed that the final state is an entangled state due to the non-separable structure of $A(j)$. The factor $A(j)$ chains  neighbour WPD states  $\ket{d_{j_k}}_k$ and $\ket{d_{j_{k+1}}}_{k+1}$ with the factor $\chi_{k+1, j_{k}, j_{k+1}}$ depending on both the paths experienced by the photon propagating through each WPD. On the other hand, each $A(j)$ is classically calculated with $\mathcal{O}(n)$ complex calculations for any $\ket{j}$. 

Next,  quantum circuit implementation of the proposed optical design is provided for simulating in recent quantum computers or NISQ devices. Quantum circuit implementation provides an easier analytical framework to explore the proposed  algorithms.
 
\section{Quantum Circuit Implementation of the Unitary Operator $U_{\vec{\Phi}}$ and Complex Hadamard Matrix Representation}
\label{sec4}
 
The optical set-up composed of multi-WPDs can be formulated and simulated with quantum circuits by modeling BSs with rotation gates \cite{amico2020simulation}. First of all, BS operator is simplified to create equal propagation probability through the paths with $\theta_k = \pi \, / \, 4$ for $k \in [1, n+1]$ resulting in the path splitting operation of $R_{X}(-\pi/2) \equiv  \dfrac{1}{\sqrt{2}}\begin{bmatrix}1 & \imath\\ \imath & 1\end{bmatrix}$. It is equal to the rotation of $-\pi \, / \, 2$ around X-axis in Bloch-sphere. It is more comparable to BS operation as compared with a Hadamard gate. Secondly, the phase shift in the upper path of each BS is set to zero with $\phi_{k,0} = 0$ for $k \in [1, n+1]$ since the difference between the phase shifts between upper and lower paths can be reflected as a global phase shift on the final quantum state.  It results in the phase shift operators defined as $\Phi_{k} \, \equiv   \, \begin{bmatrix}
1 &   0 \\
0 & e^{\imath \, \phi_{k}}\\
\end{bmatrix}$ for $k \in [1, n +1]$. Thirdly, the unitary operator on the upper path is chosen as Identity operator $I \equiv   \begin{bmatrix}1 & 0\\ 0 & 1\end{bmatrix}$ while the lower path is set to the Pauli-X operator $X \equiv   \begin{bmatrix}0 & 1\\ 1 & 0\end{bmatrix}$.  In other words, as the photon passes through the WPD, it applies a CNOT gate to the current state of the WPD. For example, if the initial state of the $k$th WPD is denoted as $\ket{D_k} = \ket{0}_k$ and the photon state after the $k$th BS is equal to $\alpha  \, \ket{0} + \, \beta \, \ket{1}$, then the entangled state of WPD and photon becomes $\alpha  \, \ket{0} \ket{0}_k + \, \beta  \, \ket{1}  \ket{1}_k$. If the initial state of the $k$th WPD is $\ket{1}_k$, then it produces $\alpha  \, \ket{0} \ket{1}_k + \, \beta  \, \ket{1}  \ket{0}_k$. Finally, we does not constrain the initial state of the photon  as $\ket{0}$ but a complex superposition of $\ket{0}$ and $\ket{1}$, i.e.,  $\alpha_0 \, \ket{0} \, + \, \alpha_{1} \, \ket{1}$, is possible. 

The quantum circuit implementation of the linear optical set-up is shown in Fig. \ref{fig2}(b). It results in the unitary operator denoted with $U_{\vec{\Phi}}$  in the family of complex Hadamard matrices  which are composed  of unit magnitude entries with arbitrary phases \cite{tadej2006concise}. These matrices are important for quantum information theory and computing with the famous example of quantum Fourier transform (QFT). In this article, a novel method is proposed to realize a specific form of complex Hadamard matrices. The formation process is analyzed by calculating the elements of $U_{\vec{\Phi}}$ for the circuit composed of phase shift gates of $\vec{\Phi} \equiv [\phi_1 \, \phi_2 \hdots \phi_{n+1} ]$. Assume that the computational basis state for the measured photon state is denoted with $\ket{e_{n+1}}$ where   $e_{n+1}$ denotes $0$ or $1$ and similarly with $\ket{d_{k}}$ (removing the symbol $k$ under the ket for simplicity) where $d_k$ denotes $k$th detector state output. Then, the elements of $U_{\vec{\Phi}}$ denoted with $U_{\vec{\Phi}}(k,l)$ are calculated as follows:
\begin{eqnarray}
U_{\vec{\Phi}}(k,l) & \equiv & \bra{e_{1}} \bra{e_{2}} \hdots \bra{e_{n+1}} \, U_{\vec{\Phi}} \, \ket{Q} \,\ket{D_n} \hdots \ket{D_{1}}  \\
& \equiv & \bra{e_{1}} \bra{e_{2}} \hdots \braket{e_{n+1} \vert \Psi_{n+1}}  
\end{eqnarray}
where $  \ket{Q} \, \ket{D_n} \hdots \ket{D_{1}}  $ denotes the initial state for the photon with $\ket{Q}$  and $k$th WPD with $\ket{D_k} $ in the computational basis states while $k = \sum_{j=1}^{n+1} e_{j} \, 2^{j-1}$ and  $l = Q\, 2^n \, + \, \sum_{j=1}^{n} D_{j} \, 2^{j-1}$.

\begin{figure*}[t!]
\includegraphics[width=6.5in]{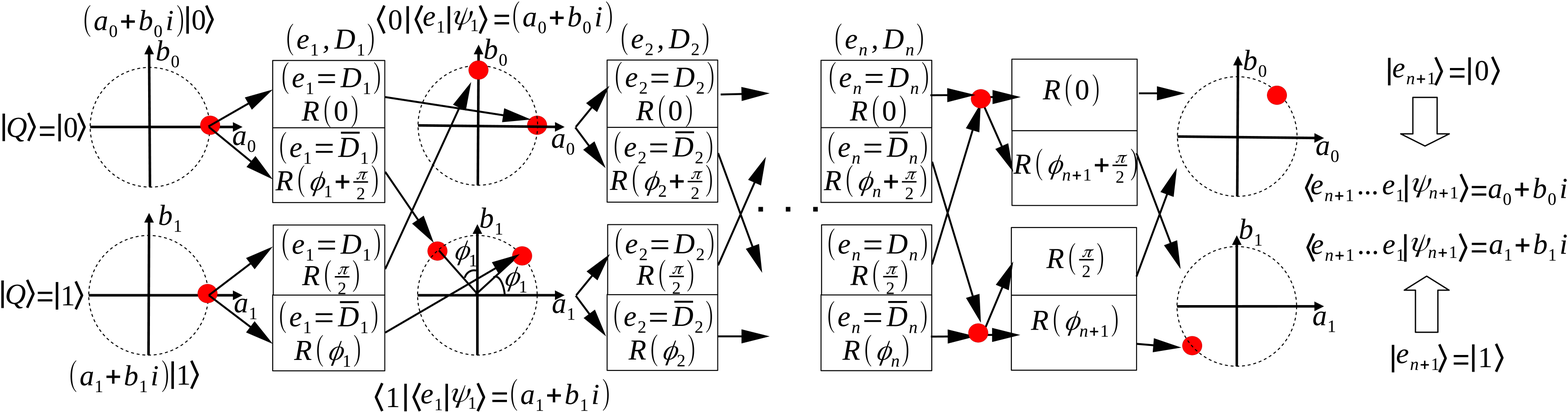}
\caption{Iterative approach to calculate $U_{\vec{\Phi}}(k,l)  \equiv   \bra{e_{1}} \bra{e_{2}} \hdots \bra{e_{n+1}} \, U_{\vec{\Phi}} \, \ket{Q} \,\ket{D_n} \hdots \ket{D_{1}}$ resulting in a special form of complex Hadamard matrix representation.}
\label{fig3}
\end{figure*}

Tracing the  quantum state,  $U_{\vec{\Phi}}(k,l)$ is calculated in classical polynomial-time in an iterative manner as described next and as shown in Fig. \ref{fig3}.   At each time step, it is assumed that the normalized unit amplitude of the photon basis state  $\ket{0}$ entangled with WPD states is denoted as $a_0 \, + \, \imath \, b_0 $ while the one for $\ket{1}$ is denoted with $a_1 \, + \, \imath \, b_1 $.  Before passing to the next step, the information about the measured state  $\ket{e_j}$ allows us to choose either $\ket{D_j}$ or $\ket{\overline{D_j}}$ for $j \in [1, n]$ to continue with in the calculation of $U_{\vec{\Phi}}(k,l)$. This operation  at the same time measures the photon state to either $\ket{0}$ or $\ket{1}$ due to entanglement with the $j$th WPD state. Therefore, at each time step, we are left with the photon state   being either $\ket{0}$ or $\ket{1}$ with the amplitudes $(a_0 \, + \, \imath \, b_0 )$ and $(a_1 \, + \, \imath \, b_1 )$. We need to track only these amplitudes depending on the pairs $D_j$ and $e_j$ for $j$th  WPD state. An iterative approach allows to obtain $U_{\vec{\Phi}}(k,l)$ with the algorithmic evaluation framework shown in Fig. \ref{fig3} where $R(\alpha) \equiv \begin{bmatrix} \cos(\alpha) & -\sin(\alpha)\\ \sin(\alpha) & \cos(\alpha)\end{bmatrix}$ denotes rotation in the counter-clockwise direction of a two dimensional vector composed of the real and imaginary components of complex amplitude of some number $z$, i.e., $z  = a_0 \, + \,  \imath \, b_0  $ and $R(\alpha) \, \begin{bmatrix}  a_0   \\ b_0 \end{bmatrix} $. After the final $(n+1)^{th}$ step, there are two possible amplitudes for the photon state which is chosen depending on $\ket{e_{n+1}}$.

For example,  at time $t_1$, $\ket{\Psi_1}$ is as follows where we reorder the position of $\ket{Q}$ to emphasize the entanglement with the first WPD state depending on the path:
\begin{eqnarray}
 \ket{\Psi_1} =  \frac{1}{2^{1/2}} \, &&   \, \ket{D_n \, \hdots \,  D_3 \,D_2} \, \otimes  \, \bigg( (a_0 \, + \, \imath \, b_0 )\,   \ket{0} \ket{D_1}   \, + \, (a_1 \, + \, \imath \, b_1 )\,  \ket{1} \ket{\overline{D_1}} \bigg) 
 \end{eqnarray} 
and the the following is obtained based on the first pair of $e_1$ and $D_1$ values:
\begin{eqnarray}
\braket{e_1 \vert \Psi_1} =  \frac{1}{2^{1/2}} \,  \, \ket{D_n \, \hdots \,  D_3 \,D_2} \, \otimes  \, \bigg( (a_0 \, + \, \imath \, b_0 )\,   \ket{0} \braket{e_1 \vert D_1}  \, + \, (a_1 \, + \, \imath \, b_1 )\,   \ket{1} \braket{e_1 \vert \overline{D_1}} \bigg)   
 \end{eqnarray} 
where $\overline{D_1}$ denotes the NOT of $D_1$ while $(a_0 \, + \, \imath \, b_0 )$ and $(a_1 \, + \, \imath \, b_1 )$ are calculated by   $R_X(-\pi/2)\, \ket{Q}$ and consecutive phase shift $\phi_1$ depending on the value of $Q$ and $e_1$. Therefore, if $e_1 = D_1$, the amplitude passed on to the next step is $(a_0 \, + \, \imath \, b_0 )$ with $\ket{0}$ while if $e_1 = \overline{D_1}$, the amplitude passed on to the next step is $(a_1 \, + \, \imath \, b_1 )$ with $\ket{1}$.  In fact, $R_{X}(-\pi \,/ \, 2) $ transforms $\ket{0}$ into $(\ket{0} \, + \, e^{\imath \, \pi \, / \, 2}\, \ket{1}) \,/ \, \sqrt{2}$ and $\ket{1}$ into $(e^{\imath \, \pi \, / \, 2} \, \ket{0} \, + \, \ket{1}) \,/ \, \sqrt{2}$ where the effect of complex $\imath$ is modeled as $\pi \, / \, 2$ rotation in the counter-clockwise direction on the complex amplitudes of $\ket{0}$ or $\ket{1}$, i.e., $R(\pi \, / \, 2)$. Therefore, if  initially $Q \, = \, 0$, i.e., $a_0 \, =  \, 1$  and $b_0 \, = \, 0$, and $e_1 \, = \, D_1$, then $R(0)$ or no change is applied on the state resulting in $a_0 \, = \,1$, $b_0 \, =  \, 0$, $a_1  \, = \, 0$, $b_1 \, = \, 0$. If initially $Q \, = \, 0$ and $e_1 \, = \, \overline{D_1}$, then both the phase factors $e^{\imath \, \pi \, / \, 2}$ and $e^{\imath \, \phi_1}$ will multiply $\ket{1}$, and the result becomes   $a_1 \,  +\, \imath \, b_1 = \, R( \phi_1 \, + \, \pi \, / \, 2)\,(a_{0} \, + \, \imath \, b_0)$ where $a_0 \, =  \, 1$  and $b_0 \, = \, 0$.  Similarly, iterative approach allows to track the phases of $\ket{0}$ and $\ket{1}$ until to $(n+1)^{\mbox{th}}$  step. In this step, $R(0)$ and $R(\phi_{n+1} \, + \, \pi / \, 2)$ are applied on $(a_0 \, + \, \imath \, b_0)$ coming from the previous step in order to calculate the amplitudes of $\ket{0}$ and $\ket{1}$ for the final step, i.e., $(a_{0} \, + \, \imath \, b_0)$ and  $ (a_1 \, + \, \imath \, b_1 )$, respectively. Similarly, $R(\pi \, / \, 2)$ and $R(\phi_{n+1})$ are applied on $(a_1 \, + \, \imath \, b_1)$ coming from the previous step  to calculate the amplitudes of $\ket{0}$ and $\ket{1}$.

The iterative approach provides a classical polynomial-time complexity formulation for the calculation of $U_{\vec{\Phi}}(k,l)$ as follows:
 \begin{equation}
U_{\vec{\Phi}}(k,l) = \frac{1}{2^{(n+1)/2}} \, \vec{u}_{e_{n+1}}^T  \, \bigg( \prod_{j = 1}^{n}\mathbf{M}_{j, D_j, e_{j}} \bigg)   \,  \vec{v}_{Q}
 \end{equation}
where $\prod_{j = 1}^{n} \mathbf{A}_{j}$ denotes the matrix product $\mathbf{A}_{n} \, \mathbf{A}_{n-1}\, \hdots \, \mathbf{A}_{1}$ from the right to the left (with the same notation in the following discussions), $k = \sum_{j=1}^{n+1} e_{j} \, 2^{j-1}$ and  $l = Q\, 2^n \, + \, \sum_{j=1}^{n} D_{j} \, 2^{j-1}$, and the following are defined:
\begin{equation}
\label{eq6and7}
    \vec{v}_{l}= 
\begin{cases}
        \begin{bmatrix}
1 &   0  &   0  &   0
\end{bmatrix}^T,               & \text{if } l \, = \, 0  \\
            \begin{bmatrix}
0 &   0  &   1  &   0
\end{bmatrix}^T,              & \text{if }  l \, = \, 1  \\
\end{cases}; \hspace{1in}
    \vec{u}_{l}= 
\begin{cases}
        \begin{bmatrix}
1 &   \imath  &   0  &   0
\end{bmatrix}^T,               & \text{if } l \, = \, 0  \\
            \begin{bmatrix}
0 &   0  &   1  &   \imath
\end{bmatrix}^T,              & \text{if }  l \, = \, 1  \\
\end{cases}
\end{equation}
\begin{equation}
\label{eq5}
 \mathbf{M}_{j, D_j, e_{j}}= 
\begin{cases}
        \begin{bmatrix}
\mathbf{I}_2  &   R(\frac{\pi}{2}) \\
\mathbf{0}_2  &   \mathbf{0}_2
\end{bmatrix} ,               & \text{if } e_{j} \, = \, D_j  \\ \\
                    \begin{bmatrix}
\mathbf{0}_2  &   \mathbf{0}_2  \\
R(\phi_j \, + \, \frac{\pi}{2} ) &   R(\phi_j)
\end{bmatrix} ,              & \text{if }   e_{j} \, = \, \overline{D}_j  \\
\end{cases}
\end{equation}
where $\mathbf{M}_{j, D_j, e_{j}}$ for $j \in [1,n-1]$ is $4 \times 4$ block diagonal matrix composed of the rotation matrix $R(\alpha) \equiv  \begin{bmatrix}
\cos(\alpha) &  -\sin(\alpha) \\
 \sin(\alpha) & \cos(\alpha)\\
\end{bmatrix} $, $\mathbf{I}_2 \equiv  \begin{bmatrix}
1 &  0 \\
0 & 1\\
\end{bmatrix} $ and $\mathbf{0}_2 \equiv  \begin{bmatrix}
0 &  0 \\
0 & 0\\
\end{bmatrix} $ and $\mathbf{M}_{n, D_n, e_{n}}$ is defined as follows:
\begin{equation}
\mathbf{M}_{n, D_n, e_{n}} = 
\begin{cases}
    \begin{bmatrix}
\mathbf{I}_2  &   \mathbf{0}_2 \\
R(\phi_{n+1} \, + \, \frac{\pi}{2} )   &   \mathbf{0}_2
\end{bmatrix}  \,     \begin{bmatrix}
\mathbf{I}_2  &   R(\frac{\pi}{2}) \\
\mathbf{0}_2  &   \mathbf{0}_2
\end{bmatrix} ,               & \text{if } e_{n} \, = \, D_n  \\ \\
                      \begin{bmatrix}
 \mathbf{0}_2 & R(\frac{\pi}{2})  \\
   \mathbf{0}_2 & R(\phi_{n+1} )   &
\end{bmatrix}  \,   \begin{bmatrix}
\mathbf{0}_2  &   \mathbf{0}_2  \\
R(\phi_n \, + \, \frac{\pi}{2} ) &   R(\phi_n)
\end{bmatrix} ,              & \text{if }   e_{n} \, = \, \overline{D}_n  \\
\end{cases}
\end{equation}
 
As a result, unit magnitude output amplitude for any input and output basis states is easily calculated with $\mathcal{O}(n)$ multiplications of complex $4 \times 4$ matrices.

\section{Eigenstructure of the Designed Unitary Operator $U_{\vec{\Phi}}$}
\label{sec5}

We order the input qubits in the reverse manner with respect to the ordering of WPDs, i.e., the bottom qubit which is the photon input state to BS is denoted with $\ket{j_0} \equiv \ket{Q}$, the initial state of the $n$th WPD as $\ket{j_1} \equiv \ket{D_n}$, and similarly for $\ket{j_{n+1-l}} \equiv \ket{D_l}$ for $l \in [1, n]$. The proposed re-ordering of the qubits is shown in Fig. \ref{fig4}. The same ordering of qubits should be kept in the output with output qubit numbers starting from $\ket{k_{0}} \equiv \ket{e_{n+1}}$ for the photon output state and $\ket{k_{n+1-l}} \equiv \ket{e_l}$  for $l \in [1, n]$.  Therefore, the elements of the unitary operator in this new re-ordered qubit basis states are calculated as follows:
\begin{eqnarray}
\label{eq1and2}
\braket{k \,\vert \,U_{\vec{\Phi}} \, \vert \, j} \,  &\equiv &\, \bra{k_{0} \,   \hdots  k_{n}}   \, U_{\vec{\Phi}} \,  \ket{j_n  \hdots     j_0}  
   =   \, \frac{\vec{u}_{k_0}^T}{2^{(n+1)/2}} \,   \bigg( \prod_{l = 1}^{n}\mathbf{M}_{l, j_{n+1-l}, k_{n+1-l}} \bigg) \,  \vec{v}_{j_0}\hspace{0.2in}
 \end{eqnarray}
where $k = \sum_{l=0}^{n} k_l \, 2^{l}$ and  $j = \sum_{l=0}^{n} j_l \, 2^{l}$. Next, the analysis and supporting evidence of Conjecture-1 are presented by showing the favorable form of the eigenvectors.

\begin{figure}[t!]
\includegraphics[width=2.5in]{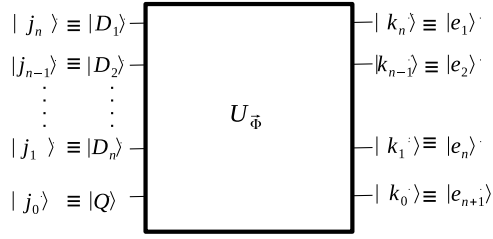}
\caption{Indexing of inputs and outputs for the quantum circuit implementation of $U_{\vec{\Phi}}$.}
\label{fig4}
\end{figure}

\subsection{Analysis and Supporting Evidence of Conjecture-1}
\label{sec5a}
Assume that input to $U_{\vec{\Phi}}$ is a state $\ket{E}_{m,\vec{s}, \vec{\alpha}}$ defined as follows:
\begin{eqnarray}
\label{eq8and9}
\ket{E}_{m,\vec{s}, \vec{\alpha}} &  \, \equiv  \, & \sum_{b=0}^{T} \, \sum_{j=0}^{M}  B_{j} \ket{j+ b\, 2^m} \\
\label{eq10}
 & \, \equiv \, &  H^{\otimes n} \ket{0 \, \hdots \, 0 \,s_{m-1} \,s_{m-2} \hdots \,s_1} (\alpha_{0} \, \ket{0}  \, + \, \alpha_{1} \, \ket{1}) \\
 \label{eq11}
& \, = \, &  \frac{1}{(\sqrt{2})^n} \, \sum_{b_{n-m}, \, \hdots, \, b_0 }  \,   \ket{b_{n-m} \hdots b_0} \,  \sum_{j_{m-1}, \, \hdots, \, j_1}    e^{\imath \, \pi \, \sum_{l=1}^{m-1} j_l \, s_l }   \ket{j_{m-1} \hdots j_1}  \sum_{j_{0}=0}^{1}  \alpha_{j_0} \, \ket{j_0}  \hspace{0.4in}
\end{eqnarray}
where the state is a separable (not entangled) state with some periodic amplitude $B_j$ of period $2^m$, $\vec{\alpha} \equiv [   \alpha_0 \,\, \alpha_1]^T$, $s_{l}$ being either $0$ or $1$ for $l \in [1, m \,- \,2]$, $s_{m-1} \,= \,1$, $M \equiv 2^{m}-1$, $T \equiv 2^{n \, - \, m \, + \, 1}-1$, $b  \,\equiv  \, \sum_{l=0}^{n-m} b_l \, 2^{l}$, $j  \,\equiv  \, \sum_{l=0}^{m-1} j_l \, 2^{l}$  and $(\alpha_0 \, \ket{0} + \alpha_1 \ket{1} )$ is the initial state of the photon. We denote $\vec{s} \equiv \left[s_{1} \, \, \hdots  \, \, s_{m-1}\right]$ as the vector of the input state. The value $s_{m-1} = 1$  is chosen in order to realize periodicity of the amplitudes $B_j$ of the basis states with $2^m$ since  the value of $s_{m-1} = 0$ will increase the number of bits in $\ket{b_{n-m} \hdots b_0}$ by one while reducing the periodicity to  $2^{m-1}$. Therefore, there are $2^{m-2}$ possible $\vec{s}$ vectors with varying binary values of $s_{l}$ for $l \, \in [1, m-2]$. In fact, the eigenvector amplitudes are discrete periodic vectors having periodicity $2^m$ for $m \in [1, n+1]$ by also including the periodicity of $2$ for all zero input vector. 
 
We target to calculate $ \braket{ k + t \, 2^m \vert   U_{\vec{\Phi}}  \, \vert E_{m,\vec{s}, \vec{\alpha}}  }$ for $k \in [0, M]$ and $t \in [0, T]$ by using (\ref{eq1and2}). Firstly, if $\ket{\overline{j}} \equiv \ket{j+ b\, 2^m}$ is defined, then it is easy to observe that $\overline{j}_l = j_l$ for $l \in [0, m-1]$ and $\overline{j}_l = b_{l-m}$ for $l \in [m, n]$ as shown in (\ref{eq8and9}-\ref{eq11}). Similarly, if $\ket{\overline{k}} \equiv \ket{k+ t\, 2^m}$ is defined, then  $\overline{k}_l = k_l$ for $l \in [0, m-1]$ and $\overline{k}_l = t_{l-m}$ for $l \in [m, n]$. Then, the following is obtained by using  (\ref{eq1and2}-\ref{eq11}):
\begin{eqnarray}
 \braket{ k + t \, 2^m \, \vert  \, U_{\vec{\Phi}}  \, \vert E_{m,\vec{s}, \vec{\alpha}} }  \equiv  
\frac{ \vec{u}_{k_0}^T}{2^{n} \, \sqrt{2}} \, \bigg(\prod_{l=1}^{m-1} \mathbf{\widetilde{M}}_{m-l, s_{m-l}, k_{m-l}} \bigg) 
\mathbf{K}_{n-m+1} \, \big(\alpha_0 \, \vec{v}_{0} \, + \, \alpha_1 \, \vec{v}_{1} \big) && 
 \end{eqnarray}
where the product matrix $\mathbf{K}_{a}$ becoming independent of $t$ is defined as follows: 
\begin{eqnarray}
\mathbf{K}_{a}  & \equiv &  \sum_{b_0=0}^{1}   \hdots    \sum_{b_{a-1}=0}^{1} \mathbf{M}_{a, b_{0}, t_{0}} \hdots \, \mathbf{M}_{2, b_{a-2}, t_{a-2}} \,  \mathbf{M}_{1, b_{a-1}, t_{a-1}} \\
& = & \prod_{l=1}^{a} \bigg( \begin{bmatrix}
\mathbf{I}_2  &   R(\frac{\pi}{2}) \\
\mathbf{0}_2  &   \mathbf{0}_2
\end{bmatrix}   +
                    \begin{bmatrix}
\mathbf{0}_2  &   \mathbf{0}_2  \\
R(\phi_l \, + \, \frac{\pi}{2} ) &   R(\phi_l)
\end{bmatrix} \bigg)
 \end{eqnarray}
and $ \mathbf{\widetilde{M}}_{a, s_a, k_a} \equiv e^{\imath \, \pi \, k_{a}\, s_{a}} \,  \mathbf{\widehat{M}}_{{a}, s_{a}}$ and  $\mathbf{\widehat{M}}_{a, s_{a}}    \equiv   \mathbf{M}_{n-a+1, 0, 0} \,+ \,  e^{\imath \, \pi \, s_{a}}  \,  \mathbf{M}_{n-a+1, 1, 0}$. Then, the resulting simplified expression is obtained:
 \begin{eqnarray}
 \label{eq13}
 \braket{ k + t \, 2^m \, \vert  \, U_{\vec{\Phi}} \, \vert  E_{m,\vec{s}, \vec{\alpha}} }  \, =  \, e^{\imath \, \pi \, \sum_{l=1}^{m-1} k_l \, s_l } \sum_{j_0=0}^1 \big( \varrho_{j_0, k_0} \, \alpha_{j_0} \big)    \hspace{0.2in}
 \end{eqnarray}
 where $\varrho_{j_0, k_0} \equiv   \vec{u}_{k_0}^T   \,\mathbf{V}_{\vec{s}} \,  \vec{v}_{j_0} \, / \,  (2^{n} \, \sqrt{2})$ for $j_0 =0$ or $1$  and $\mathbf{V}_{\vec{s}} \equiv  \big(\prod_{l=1}^{m-1} \mathbf{\widehat{M}}_{m-l, s_{m-l}} \big)  \, \mathbf{K}_{n-m+1} $.  Observe that $\mathbf{V}_{\vec{s}}$ is easily calculated with $\mathcal{O}(n)$ multiplications of complex $4 \times 4 $ matrices  where the matrices depend on $\vec{\Phi}$, i.e., the phase shift values in each step.

If  there exists a solution (depending on the provided values $\phi_k$ for $k \in [1, n+1]$) for  $\alpha_0$, $\alpha_1$ and the eigenvalue $e^{\imath \, \lambda_{m, \vec{s}, \vec{\alpha}}}$  by solving the equation $  \sum_{j_0=0}^1 \big( \varrho_{j_0, k_0} \, \alpha_{j_0} \big) =  e^{\imath \, \lambda_{m, \vec{s}, \vec{\alpha}}} \, \alpha_{k_0} \, / \, (\sqrt{2})^n$ for $k_0 = 0$ and $1$ simultaneously, then by using (\ref{eq11}) and (\ref{eq13})  it is observed that $\ket{E}_{m,\vec{s}, \vec{\alpha}}$ becomes  an eigenvector of $U_{\vec{\Phi}}$ as follows: \begin{eqnarray}
U_{\vec{\Phi}} \, \ket{E}_{m,\vec{s}, \vec{\alpha}} & = &  e^{\imath \, \lambda_{m, \vec{s}, \vec{\alpha}}} \, \ket{E}_{m,\vec{s}, \vec{\alpha}} 
 \end{eqnarray}
where the eigenvalue is equal to the following:
\begin{equation}
\label{eq14}
e^{\imath \, \lambda_{m, \vec{s}, \vec{\alpha}}} \,  = \, \frac{\varrho_{0, 0} \, \alpha_0 + \varrho_{1, 0} \, \alpha_1}{\alpha_0 \, / \, (\sqrt{2})^n}  =\frac{\varrho_{0, 1} \, \alpha_0 + \varrho_{1, 1} \, \alpha_1}{\alpha_1 \, / \, (\sqrt{2})^n}
\end{equation} 

We can constrain $\alpha_0$ and $\alpha_1$ further by setting one of them as real, e.g., similar to the Bloch sphere representation by omitting a global phase, i.e., $\alpha_0 \ket{0} \, + \, \alpha_1 \, \ket{1} \equiv \cos(\varphi \, / \, 2)  \, \ket{0} \, + \, e^{\imath \, \gamma} \sin(\varphi \, / \, 2)  \, \ket{0}$. Assume that $\alpha_0 = a_0 \, + \, \imath \, b_0 $ and $\alpha_1 = a_1 \, + \, \imath \,b_{1}$ where either $b_0$ or $b_1$ equal to zero. Another constraint is that the eigenvector in (\ref{eq11}) should be normalized so that $\vert \alpha_0 \vert^2 \, + \,\vert \alpha_1 \vert^2 = 1$. 

Then, the four unknowns, i.e., $a_0$, $b_0$, $a_1$ and $b_1$ can be solved with the four polynomial equalities with one of them having complex coefficients. The first equality is due to  $\vert \sum_{j_0=0}^1 \big( \varrho_{j_0, k_0} \, \alpha_{j_0} \big) \vert^2$ being equal to $\vert \alpha_{k_0} \vert ^2 \, / \, 2^n$ for $k_0 \, = \, 0$ or $1$:
\begin{eqnarray}
\label{eq15}
\vert \varrho_{0, k_0} \, (a_0   +  \, \imath \, b_0) + \varrho_{1, k_0} \, \big(a_1  \, +  \,\imath \, b_1\big) \vert^2   & = & (a_{k_0}^2 \,  + \, b_{k_0}^2) \, / \, 2^n  \hspace{0.4in}  
\end{eqnarray}
The second equality is due to the expanded form of (\ref{eq14}) by excluding the eigenvalue as follows:
\begin{eqnarray}
\label{eq16}
(a_1   + \, \imath \, b_1) \big( (a_0   +   \imath \,  b_0) \varrho_{0, 0} +     (a_1   +   \imath \, b_1)   \varrho_{1, 0} \big)  
-(a_0   +   \imath \, b_0) \big( (a_0   + \, \imath \,  b_0)   \varrho_{0, 1}   +  (a_1  +   \imath  \, b_1)   \varrho_{1, 1} \big)  =  0  \hspace{0.4in}
\end{eqnarray}
Finally, the third and fourth equalities are as follows:
\begin{eqnarray}
\label{eq17}
 a_0^2 \, + \, b_0^2 + \, a_1^2 \, + \, b_1^2 & \, = \, & 1 \\
 \label{eq18}
  b_0 \, b_1 &\, = \, & 0 
\end{eqnarray}

It can be observed that if $b_1 = 0$, and $(\alpha_0  = a_0 \, + \, \imath \, b_0, \alpha_1 = a_1)$ is a solution to the equations, then $(\alpha_0 = a_1, \alpha_1 = - a_0 \, + \, \imath \, b_0)$ (or its multiplication with some phase) becomes also a solution by creating a duality for the single photon state for the same $\vec{s}$. The trigonometric proof is provided in Appendix \ref{appB} with some minor open issues.  

Observe that $\mathbf{V}_{\vec{s}}$ includes multiplications of $4 \times 4$ matrices composed of  $2 \times 2$ block matrices applying rotations depending on $\phi_k$ for $k \in [1, n+1]$.   It is an open issue to determine the conditions on $\phi_k$ for $k \in [1, n+1]$ under which there exist solutions satisfying four equalities (\ref{eq15}-\ref{eq18}) and the duality condition with all different eigenvalues  so that the conditions and assumptions in Theorem-1 and consecutively Conjecture-1 are satisfied.  In fact, it can be hypothesized that uniformly distributed values of phase shifts result in different eigenvalues satisfying the proposed assumptions based on extensive numerical analysis.

\section{Open Problems}
\label{sec6}
\begin{itemize}
\item The existence of the specific eigenstructure form of $U_{\vec{\Phi}}$ definitely depends on the chosen  phase shift values $\phi_k$ for $k \in [1, n+1]$. It is an open issue under which conditions it provides such a favorable eigenstructure to be utilized in Theorem-1. In numerical analysis, it is observed that uniformly distributed random values of $\phi_k$ for $k \in [1, n+1]$ produce such operators $U_{\vec{\Phi}}$  having the desired structure by providing numerical evidence.
\item Quantum polynomial-time complexity solution in Theorem-1 depends on the  existence  of black-boxes or oracles, i.e., controlled-$U_{\vec{\Phi}}^{2^j}$ operations  for $j \in [0, \widetilde{t} \, - \,1]$. We have provided a design of $U_{\vec{\Phi}}$ based on linear optics and WPDs while also providing polynomial size quantum circuit implementation. It is an open issue how to implement exponentially large powers of the designed $U_{\vec{\Phi}}$  with polynomial size quantum circuits.  Is there any specific set of phase shift values $\phi_k$ for $k \in [1, n+1]$ so that both Conjecture-1 is satisfied and polynomial size quantum circuits exist for implementing $U_{\vec{\Phi}}^{2^j}$ based on the proposed linear optical design?
\item Is there any other practical design of $U_{\vec{\Phi}}$ having the desired eigenstructure so that its exponentially large powers can be implemented with polynomial size quantum circuits? 
\item What are the effects of noise in input state and how are the estimation errors modeled? 
\item How can the proposed unitary operator design be exploited in QST of mixed state inputs?
\item The designed operator $U_{\vec{\Phi}}$ has eigenvalues and eigenvectors to be calculated easily by using multiplications of matrix factorizations as shown in Section \ref{sec5}. It is conjectured as a one-way function promising to be utilized in various cryptographic algorithms as an open issue. 
\end{itemize}

\begin{acknowledgments}
This work was supported by TUBITAK (The Scientific and Technical Research Council of Turkey) under Grant $\#$119E584.
\end{acknowledgments}

\appendix

\section{Evolution of single photon state in linear optical set-up with multiple WPDs}
\label{appA}
We assume that  $BS_k$ in Fig. \ref{fig2}(a) applies the following operator:
\begin{equation} 
\label{eq3}
U_{BS,k} \, = \, \begin{bmatrix}
\cos(\theta_k) & \imath \, \sin(\theta_k) \\
\imath \, \sin(\theta_k) & \cos(\theta_k)\\
\end{bmatrix} 
\end{equation}
The standard phase shift operator shifts the phase of the photon in the path with the following for the chosen path index $j = 0$ or $1$:
\begin{equation} 
\label{eq4}
\Phi_{k,j} \, \equiv   \, \begin{bmatrix}
1 &   0 \\
0 & e^{\imath \, \phi_{k,j}}\\
\end{bmatrix} 
\end{equation}

WPD with the index $k$ applies either $U_{k,0}$ or  $U_{k,1}$ depending on the selected path of the photon passing through BS. Although there is no constraint on the unitaries and the dimension of Hilbert space in WPDs, we assume that each WPD state is two dimensional with the corresponding single qubit unitary gates $U_{k,0}$ and $U_{k,1}$. We assume that the initial states in WPDs at time $t_0$ are pure states represented by $\ket{D_n} \, \otimes \, \ket{D_{n-1}} \, \hdots \otimes  \ket{D_1}$ in order to better observe the evolution of each path through consecutive WPDs.  At time $t_i$ the state is denoted as $\ket{\Psi_i}$.   Then at time $t_0$, the initial state is as follows:
\begin{equation}
\ket{\Psi_0} \, = \,\ket{Q} \, \otimes \, \ket{D_n} \, \otimes \, \ket{D_{n-1}} \, \hdots \, \otimes \, \ket{D_1}
\end{equation}
where $\ket{Q} = \alpha_0 \, \ket{0} \, + \, \alpha_1 \, \ket{1}$ denotes the initial state of the single photon.  Assume that $\ket{Q} = \ket{0}$ to analyze the evolution of the WPD states more simply.  In the following, we remove the tensor product $\otimes$ in order to simplify the expressions which can be understood from the context. Furthermore, we reorder the positions of the WPD states for the following modeling in order to better characterize the evolution since the interaction of WPDs starts with the first WPD. At the end, we correct again.

After the first beam splitter $BS_1$, the state is transformed into the following:
\begin{equation}
  U_{BS,1} \, \ket{0} \,\ket{D_1}   \, \ket{D_{2}} \, \hdots   \, \ket{D_n}  
\end{equation}
The state becomes  $\big(\cos(\theta_1) \ket{0} \, + \, \imath \, \sin(\theta_1) \, \ket{1}\big)\, \ket{D_1   \,  D_{2} \, \hdots   \,  D_n}   $. Then, unitary operator is applied on the first WPD state depending on the path chosen by the photon and the state becomes entangled with the photon path as follows:
\begin{equation}
 \big(\cos(\theta_1) \ket{0}  U_{1,0} \ket{D_1} \, + \, \imath \, \sin(\theta_1)  \ket{1}  U_{1,1} \ket{D_1} \big)\,\ket{  D_2  \, \hdots   \,  D_n }  \
\end{equation}
Finally, phase shift operator $\Phi_{1,j}$ is applied depending on the path resulting in $\ket{\Psi_1}$ as a superposition of two paths entangling the photon and the first WPD state:
\begin{eqnarray}
\ket{\Psi_1} \,   = \,  \big( \cos(\theta_1) \, e^{\imath \, \phi_{1,0}}  \, \ket{0}  U_{1,0} \ket{D_1}   \, + \, \imath \, \sin(\theta_1)  \, e^{\imath \, \phi_{1,1}}  \, \ket{1}  U_{1,1} \ket{D_1} \big) \,  \ket{ D_2  \, \hdots   \,  D_n }   
\end{eqnarray}
Defining $\widetilde{c}_{k,j} \equiv \cos(\theta_k) \, e^{\imath \, \phi_{k,j}}$, $\widetilde{s}_{k,j} \equiv \imath \, \sin(\theta_k)  \, e^{\imath \, \phi_{k,j}}  $ and $ \ket{d_j}_{k} \equiv U_{k,j} \ket{D_k}$ simplifies the expression as $\ket{\Psi_1}    =    \,   \big(  \widetilde{c}_{1,0}  \, \ket{0} \ket{d_0}_{1}   +    \widetilde{s}_{1,1} \, \ket{1}   \ket{d_1}_{1} \big) \, \ket{ D_2  \, \hdots   \,  D_n }$. Similarly, the entanglement of the photon with the next WPD detector results in  $\ket{\Psi_2}$ consisting of the superposition of four paths:
\begin{eqnarray}
\ket{\Psi_2} = && \big( \ket{0}( \widetilde{c}_{2,0}\, \widetilde{c}_{1,0} \, \ket{d_0}_{1}  \, + \, \widetilde{s}_{2,0}\, \widetilde{s}_{1,1} \, \ket{d_1}_{1} ) \ket{d_0}_{2} \, + \ket{1}( \widetilde{c}_{1,0}\, \widetilde{s}_{2,1} \, \ket{d_0}_{1}  \, + \, \widetilde{c}_{2,1}\, \widetilde{s}_{1,1} \, \ket{d_1}_{1} ) \ket{d_1}_{2} \big)  \nonumber \\
&&  \, \ket{ D_3  \, \hdots   \,  D_n } \,  \hspace{0.2in}
\end{eqnarray}
It can be simply iterated until the $(n+1)$th step to compute $\ket{\Psi_{n+1}}$ as follows after reordering the positions of the WPD states:
\begin{eqnarray}
\sum_{j_{1} =0}^{1} \,\sum_{j_{2} =0}^{1} \, \hdots \, \sum_{j_{n} =0}^{1} \,   \sum_{j_{n+1} =0}^{1} \,  A(j_1, \hdots, j_n, j_{n+1})    \,  \, \ket{j_{n+1}}_{n+1} \, \ket{d_{j_n}}_n \, \hdots \, \ket{d_{j_1}}_1\hspace{0.15in}  &&  
\end{eqnarray}
where we denote the final state of the photon as $\ket{j_{n+1}}_{n+1}$ and $A(j_1, \hdots, j_n, j_{n+1})$ is defined as follows:
\begin{eqnarray}
A(j_1, \hdots,  j_{n+1}) & \equiv & K_{j_1} \prod_{k=1}^{n} \chi_{k+1, j_{k}, j_{k+1}} 
\end{eqnarray}
where the following amplitudes are defined:
\begin{eqnarray}
    \chi_{k+1, j_{k}, j_{k+1}} & = & 
\begin{cases}
    \widetilde{c}_{k+1, j_{k+1}},               & \text{if } j_{k} \, = \, 0, \, j_{k+1}\,=\,0 \\
    \widetilde{s}_{k+1, j_{k+1}},              & \text{if } j_{k} \, = \, 1, \, j_{k+1}\,=\,0 \\
    \widetilde{s}_{k+1, j_{k+1}},               & \text{if } j_{k} \, = \, 0, \, j_{k+1}\,=\,1 \\
    \widetilde{c}_{k+1, j_{k+1}},              & \text{if } j_{k} \, = \,1, \, j_{k+1}\,=\,1 \\
\end{cases}\\
    K_{j_1} & = & 
\begin{cases}
    \widetilde{c}_{1,0},               & \text{if } j_{1} \, = \, 0  \\
    \widetilde{s}_{1,1},              & \text{if } j_{1} \, = \, 1 \\
 \end{cases}
\end{eqnarray}

\section{Duality of the solution for  ancillary qubit state}
\label{appB}

Assume that the following equalities are satisfied for some unknown $\omega$, $\gamma_{0, 0}$,  $\gamma_{0, 1}$  and  $\gamma_{1, 0}$ depending on $\phi_k$ for $k \in [1, n+1]$ for varying $\vec{s}$ as discussed in Conjecture-1 in Section \ref{sec5a}: 
\begin{eqnarray}
\label{appBeq1}
 \varrho_{0,0} & \, = \,& \big(1  \, / \, (\sqrt{2})^n \big) \, \cos(\omega) \, e^{\imath \, \gamma_{00}} \\
 \label{appBeq2}
  \varrho_{1,0} & \, = \,& \big(1  \, / \, (\sqrt{2})^n \big) \, \sin(\omega)\, e^{\imath \, \gamma_{10}} \\
  \label{appBeq3}
   \varrho_{0,1} & \, = \,& \big(1  \, / \, (\sqrt{2})^n \big) \, \sin(\omega)\, e^{\imath \, \gamma_{01}} \\
   \label{appBeq4}
    \varrho_{1,1} & \, = \,& \big(1  \, / \, (\sqrt{2})^n \big) \, \cos(\omega)\, e^{\imath \, (\gamma_{01} \, + \, \gamma_{10} \,- \, \gamma_{00}\,- \, \pi)} 
\end{eqnarray}
where $\varrho_{j_0, k_0}$ are expressed for $j_0, \, k_0 \, \in [0,1]$. Furthermore, assume that the equalities (\ref{eq15}-\ref{eq18}) have a solution with $b_1 \, = \,0$, $a_0 =  \mu_0$, $b_0 = \mu_1$ and $a_1 = \mu_2$. Then, by inserting the solutions and using the assumptions in (\ref{appBeq1}-\ref{appBeq4}), the satisfied equalities in (\ref{eq15}-\ref{eq18})  are converted to the following: 
\begin{eqnarray}
E_1  \, \equiv \, &&   \mu_2 \,\sin (2 \,  \omega) \big(\mu_0 \cos (\gamma_{00} \,- \,\gamma_{10}) \, - \, \mu_1 \sin (\gamma_{00} \, - \, \gamma_{10}) \big) \,- \, \sin ^2(\omega) \big( \mu_0^2 \, + \, \mu_1^2 \, - \, \mu_2^2 \big)    \\
E_2  \, \equiv \,&&    \sin ^2(\omega) \big(\mu_0^2 \,+ \, \mu_1^2 \,- \, \mu_2^2 \big) \, + \, \mu_2 \, \sin (2 \,\omega) \big(\mu_1  \, \sin (\gamma_{00} \, - \, \gamma_{10}) \, - \, \mu_0  \, \cos (\gamma_{00} \, - \, \gamma_{10}) \big)  \\
E_3  \, \equiv \, && \sin (\omega) \Big( \big( \mu_1^2 \,- \, \mu_0^2 \big) \cos (\gamma_{00} \, + \, \gamma_{01}) \, + \, 2  \, \mu_0 \, \mu_1 \, \sin (\gamma_{00} \, + \, \gamma_{01}) \, + \, \mu_2^2 \, \cos (\gamma_{00} \,+ \, \gamma_{10})\Big)  \nonumber \\ 
&& +  \, \mu_2 \,  \cos (\omega) \, \Big( \mu_0 \big(\cos (2 \, \gamma_{00}) \, + \, \cos (\gamma_{01} \, + \, \gamma_{10}) \big) \, - \, \mu_1 \, \big(\sin (2 \, \gamma_{00}) \, + \, \sin (\gamma_{01} \,+ \, \gamma_{10}) \big) \Big)  \\
E_4  \, \equiv \, && \sin (\omega)  \Big( \big( \mu_1^2 \, - \, \mu_0^2 \big) \sin (\gamma_{00} \,+ \, \gamma_{01}) \, - \, 2  \, \mu_0 \, \mu_1  \, \cos (\gamma_{00} \, + \, \gamma_{01}) \, + \, \mu_2^2  \, \sin (\gamma_{00} \, + \, \gamma_{10}) \Big)  \nonumber \\
&& + \mu_2 \, \cos (\omega) \Big( \mu_0 \big(\sin (2  \, \gamma_{00}) \,+ \, \sin (\gamma_{01} \,+ \, \gamma_{10}) \big) \, + \, \mu_1  \, \big( \cos (2 \gamma_{00}) \, + \, \cos (\gamma_{01} \, + \, \gamma_{10}) \big) \Big)  \hspace{0.2in}\\
E_5 \, \equiv \, && \mu_0^2 + \mu_1^2 + \mu_2^2 - 1   
\end{eqnarray}
where $E_k \, = \, 0$ for $k \in [1, 5]$. Similarly, if  another solution, i.e., the dual solution, is $b_0 \, = \,0$, $a_0 =  \mu_2$, $b_1 = \mu_1$ and $a_1 = - \, \mu_0$, then (\ref{appBeq1}-\ref{appBeq4}) are converted to the following which should give $E_k \, = \, 0$ for $k \in [6, 10]$:
\begin{eqnarray}
E_6  \, \equiv \, &&   - \, E_1  \\
E_7  \, \equiv \,&&   - \, E_2  \\
E_8  \, \equiv \, && \sin (\omega)   \Big( \big( \mu_1^2 \, - \, \mu_0^2 \big)  \cos (\gamma_{00} \, + \, \gamma_{10}) \, - \, 2 \, \mu_0 \,  \mu_1  \, \sin (\gamma_{00} \, + \, \gamma_{10}) \, + \, \mu_2^2 \cos (\gamma_{00} \, + \, \gamma_{01}) \Big) \nonumber \\
&& + \, \mu_2 \cos (\omega) \Big( \mu_0  \, (\cos (2 \gamma_{00}) \,+ \, \cos (\gamma_{01} \, + \, \gamma_{10})) \, + \, \mu_1 \big( \sin (2 \,\gamma_{00}) \, + \, \sin (\gamma_{01} \, + \,\gamma_{10}) \big) \Big) \\
E_9  \, \equiv \, && \sin (\omega)  \Big( \big( \mu_1^2 \, - \, \mu_0^2 \big)  \sin (\gamma_{00} \, + \, \gamma_{10}) \, + \, 2  \, \mu_0 \,  \mu_1 \,  \cos (\gamma_{00} \, + \, \gamma_{10}) \, + \, \mu_2^2 \, \sin (\gamma_{00} \, + \, \gamma_{01}) \Big) \nonumber \\
 && + \, \mu_2 \cos (\omega) \Big( \mu_0 (\sin (2 \gamma_{00}) \, + \, \sin (\gamma_{01} \, + \, \gamma_{10})) \, - \, \mu_1  \big( \cos (2 \, \gamma_{00}) \, + \, \cos (\gamma_{01} \, + \, \gamma_{10})  \big) \Big)  \hspace{0.2in}\\
E_{10} \, \equiv \, && E_5
\end{eqnarray}
Therefore, for both the solutions to exist, $E_k$ should be zero for $k \in [1, 10]$. Equating $E_{1}$, $E_2$ and $E_3$ to zero gives the following solutions for $\omega$ and $\mu_0$:
\begin{eqnarray}
\omega \, & = & \, \tan ^{-1}\Bigg(-\, \frac{\mu_2 \,\Big( \mu_0 \big(\cos (2 \, \gamma_{00}) \, + \, \cos (\gamma_{01} \,+ \, \gamma_{10}) \big) \, - \, \mu_1 \, \big(\sin (2 \, \gamma_{00}) \, + \, \sin (\gamma_{01} \, + \, \gamma_{10}) \big) \Big)}{\left(\mu_1^2 \,- \, \mu_0^2\right) \cos (\gamma_{00} \, + \, \gamma_{01}) \, + \, 2 \,  \mu_0 \, \mu_1  \, \sin (\gamma_{00} \, + \, \gamma_{01}) \,+ \, \mu_2^2 \, \cos (\gamma_{00} \, + \, \gamma_{10})} \bigg) \hspace{0.4in} \\
\mu_0\, & = & \, -  \, \mu_1 \, \cot \left(\frac{\gamma_{01} \,- \, \gamma_{10}}{2}\right) 
\end{eqnarray} 
Substituting $\mu_0$ into $E_k$ for $k \in \lbrace  3, 4, 8, 9\rbrace$ results in the following equalities:
\begin{eqnarray}
\frac{E_3}{\cos(\gamma_{00} \, + \, \gamma_{10})} \, - \, \frac{E_8}{\cos(\gamma_{00} \, + \, \gamma_{01})} & = &  0 \\
\frac{E_4}{\sin(\gamma_{00} \, + \, \gamma_{10})} \, - \, \frac{E_9}{\sin(\gamma_{00} \, + \, \gamma_{01})} & = &  0 
\end{eqnarray}

If $\cos(\gamma_{00} \, + \, \gamma_{01})$, $\cos(\gamma_{00} \, + \, \gamma_{10})$, $\sin(\gamma_{00} \, + \, \gamma_{01})$ and $\sin(\gamma_{00} \, + \, \gamma_{10})$ are not equal to zero, then the equalities $E_3 = E_4 = E_8 = E_9 = 0$ are  satisfied together. It is also easily seen that $E_1 \, + \, E_2 = 0$, and since $E_6 \, = \, - \, E_1$, $E_7 \, = \, - \, E_2$ and $E_{10} = E_5$, the equalities $E_{1} = E_2 = E_5 = E_6 = E_ 7 = E_{10}= 0$ are also satisfied together. Therefore, all the equalities are satisfied for being equal to zero with some conditions on $\gamma_{00}$,  $\gamma_{01}$ and $\gamma_{10}$ which depend on $\phi_k$ for $k \in [1, n+1]$ for varying $\vec{s}$. 

We have utilized the assumptions about the forms of $\varrho_{k_0, j_0}$ in (\ref{appBeq1}-\ref{appBeq4}). In order to satisfy these equalities, $E_{k}$ values should be equal to zero for $k \in [11, 14]$  defined as follows:

\begin{eqnarray}
E_{11}  \, \equiv \, &&   \vert \varrho_{0,0} \vert ^2  \, + \,  \vert \varrho_{1,0} \vert ^2  \,  -\,  1 \  / \, 2^n \\
E_{12}  \, \equiv \,&&   \vert \varrho_{0,0} \vert ^2  \, + \,  \vert \varrho_{0,1} \vert ^2  \,  -\,  1 \  / \, 2^n  \\
E_{13}  \, \equiv \, &&   \vert \varrho_{0,1} \vert ^2  \, + \,  \vert \varrho_{1,1} \vert ^2  \,  -\,  1 \  / \, 2^n  \\
E_{14}  \, \equiv \,&&   \varrho_{0,0}^* \, \varrho_{1,0} +  \varrho_{0,1}^* \, \varrho_{1,1} 
\end{eqnarray}

$E_{k}$ for $k \in [11, 14]$  are calculated by using $\varrho_{j_0, k_0} \equiv   \vec{u}_{k_0}^T   \, \big(\prod_{l=1}^{m-1} \mathbf{\widehat{M}}_{m-l, s_{m-l}} \big)  \, \mathbf{K}_{n-m+1} \,  \vec{v}_{j_0} \, / \,  (2^{n} \, \sqrt{2})$  with the modeling in Section \ref{sec5a}. If the most general case is chosen with $m = n+1$, then  $\varrho_{j_0, k_0}$ becomes equal to $   \vec{u}_{k_0}^T   \, \big(\prod_{l=1}^{n} \mathbf{\widehat{M}}_{n+1-l, s_{n+1-l}} \big)    \,  \vec{v}_{j_0} \, / \,  (2^{n} \, \sqrt{2})$. One can perform the calculations with symbolic variables $\phi_k$ to check whether $E_{k} \, = \, 0$ for $k \in [11, 14]$. We have performed iterations of symbolic calculations until $n \, < \, 7$ for practical values of $n$ to observe that  $E_{k} \, = \, 0$ for $k \in [11, 14]$. It is a minor open issue to show that the equalities are satisfied for all $n$ by using trigonometry and by exploiting the rotation matrices in the definition of $\mathbf{\widehat{M}}_{k, s_{k}}$ for a more rigorous proof.

As a result, there exist  dual solutions depending on whether $\phi_k$ for $k \in [1, n+1]$ satisfies $E_k = 0$ for $k \in [1, 10]$ while $\gamma_{00} \, + \, \gamma_{01}$ and $\gamma_{00} \, + \, \gamma_{10}$ are not multiples of $ \pi   \, / \, 2$. It is an open issue to determine the class of $\phi_k$ for $k \in [1, n+1]$  providing the favorable solution and the eigenstructure. Uniformly distributed values of $\phi_k$ for $k \in [1, n+1]$ provide promising solutions in numerical analysis as a supporting numerical evidence.

\newpage
\nocite{*}

\bibliography{QIP2022_TechnicalManuscript_v7} 

\end{document}